\documentclass[a4paper,UKenglish,cleveref, autoref, thm-restate]{lipics-v2021}

\usepackage[utf8]{inputenc}
\nolinenumbers
\title{Satisfiability Checking of Multi-Variable TPTL with Unilateral Intervals is PSPACE-complete.}

\author{Shankara Narayanan Krishna}{IIT Bombay, Mumbai, India}{krishnas@cse.iitb.ac.in}{}{}

\author{Khushraj Nanik Madnani}{MPI-SWS, Kaiserslautern, Germany}{kmadnani@mpi-sws.org}{https://orcid.org/0000-0003-0629-3847}{}

\author{Rupak Majumdar}{MPI-SWS, Kaiserslautern, Germany}{rupak@mpi-sws.org}{}{}

\author{Paritosh K. Pandya}{IIT Bombay, Mumbai, India}{pandya@tifr.res.in}{}{}

\authorrunning{S.N.Krishna, K.N.Madnani, R. Majumdar, P.K. Pandya}

\titlerunning{Satisfiability Checking for $\tptlu$ is PSPACE-Complete.}

\Copyright{Shankara Narayanan Krishna,
Khushraj Madnani,
Rupak Majumdar and Paritosh K. Pandya} 


\funding{
This research was sponsored in part by
the Deutsche Forschungsgemeinschaft project 389792660 TRR 248—CPEC.
}
 
\usepackage[normalem]{ulem}
\usepackage{multirow}
\usepackage{amssymb}
\usepackage{verbatim}
\usepackage{cancel}
\usepackage{amsmath,amsfonts}
\usepackage[section]{placeins}
\usepackage{lineno}
\usepackage{expl3}
\ExplSyntaxOn
\cs_new_eq:NN \Repeat \prg_replicate:nn
\ExplSyntaxOff
\usepackage{graphicx}
\usepackage{wrapfig}
\usepackage[disable]{todonotes}
\usepackage{color}  
\usepackage{xcolor}
\usepackage{amsmath}
\usepackage{amsfonts}
\usepackage{amssymb}
\usepackage{tikz}

\usetikzlibrary{automata}
\usepackage{mathtools}

\renewcommand{\r}{\mathcal{R}}

\usetikzlibrary{arrows,automata,chains,matrix,positioning,scopes,calc,decorations.pathmorphing,shapes.misc,shapes,fit}
\tikzset{mode/.style={font=\scriptsize}}
\tikzset{gadget/.style={->,>=stealth,initial text=,minimum size=7pt,auto,on grid,scale=1,inner sep=1pt,node distance=1cm}}
\tikzset{every state/.style={minimum size=15pt,inner sep=1pt,fill=black!10,draw=black!70,thick}}

\usetikzlibrary{calc}
\tikzstyle{RectObject}=[rectangle,fill=white,draw,line width=0.5mm]
\tikzstyle{line}=[draw]
\tikzstyle{arrow}=[draw, -latex]

\newcommand{\rhos}{\mathbf{\rho}}

\newcommand{\rs}{\mathbf{R}}

\newcommand{\buchi}{B\"{u}chi }

\newcommand {\Gg} {\mathcal{G}^{\ge}}
\newcommand {\Gl} {\mathcal{G}^{\le}}
\newcommand{\Nu}{\mathcal{V}}

\newcommand{\sbf}{\mathsf{G}}

\newcommand{\fut}{\mathsf{F}}

\newcommand{\until}{\mathsf{U}}
\newcommand{\since}{\mathsf{S}}

\newcommand{\nex}{\mathsf{Next}}
\newcommand{\nx}{\nex}

\newcommand{\tred}{\mathsf{T}_{red}}
\newcommand{\act}{\mathsf{Active}}
\newcommand{\reset}{\mathsf{Reset}}
\newcommand{\ac}{\mathsf{Act}}
\newcommand{\appear}{\mathsf{Appear}}
\newcommand{\Aa}{\mathcal{A}}

\newcommand{\mtl}{\text{MTL}}
\newcommand{\tptl}{\text{TPTL}}
\newcommand{\mitl}{\text{MITL}}
\newcommand{\mitlu}{\text{MITL}^{0,\infty}}
\newcommand{\tptlu}{\text{TPTL}^{0,\infty}}
\newcommand{\ata}{\text{ATA}}
\newcommand{\atau}{{\text{ATA}^{0,\infty}}}
\newcommand{\vwata}{\text{VWATA}}
\newcommand{\vwatau}{{\vwata^{0,\infty}}}
\newcommand{\emitl}{\text{EMITL}}

\newcommand{\ltl}{\mathsf{LTL}}

\newcommand{\cl}{\mathsf{cl}}

\newcommand{\formula}{\mathsf{frm}}
\newcommand{\Cac}{$\mathcal{C}_{acc}$ }

\newcommand{\Img}{\mathsf{Img}}
\newcommand{\type}{\mathsf{Type}}
\newcommand{\ca}{\cong_{\Aa}}
\newcommand{\la}{\le_{\Aa}}
\newcommand{\lasim}{\preceq}
\newcommand{\reduce}{\mathsf{Red_{\lasim}}}
\newcommand{\Qg}{\mathsf{Q^{\ge}}}
\newcommand{\Ql}{\mathsf{Q^{\le}}}
\newcommand{\Qsim}{\mathsf{Q^{\sim}}}
\newcommand{\Q}{\mathcal{Q}}
\newcommand{\Qc}{\mathcal{Q}_{acc}}
\newcommand{\q}{\mathsf{q}}
\newcommand{\R}{\mathbb{R}}

\newcommand{\A}{\mathbb{A}}

\newcommand{\intintervalz}{\intinterval^{0}}
\newcommand{\intintervalinf}{\intinterval^{\infty}}
\newcommand{\intintervalu}{\intinterval^{0,\infty}}
\newcommand{\newclass}[2]{\newcommand{#1}{\textsc{#2}}}
\newclass{\pspace}{PSpace }
\newclass{\expspace}{ExpSpace }

\newcommand{\Qacc}{\mathsf{Q_{acc}}}

\newcommand{\true}{\mathsf{True}}
\newcommand{\false}{\mathsf{False}}

\newcommand{\modelmin}{\models^{\mathsf{min}}}
\renewcommand{\succ}{\mathsf{Succ}_\delta}
\newcommand{\suca}{\mathsf{Succ}_{\Aa}}

\newcommand{\gap}{\;}

\newcommand{\X}{\mathcal{X}}

\newcommand{\C}{\mathbb{C}}
\newcommand{\G}{\mathcal{G}}

\newcommand{\intintervaln}{\mathcal{I}_{\mathsf{int}}}

\newcommand{\intinterval}{\mathcal{I}_{\mathsf{int}}}

\newcommand{\lo}{\mathcal{L}}

\newcommand{\oomit}[1]{}
\newcommand{\init}{\text{init}}

\newcommand{\cntphi}{\mathsf{Count}\phiapproach}
\newcommand{\phiapproach}{\phi_{\mathsf{approach}}}
\usetikzlibrary{shapes,arrows,positioning,trees, automata}

\begin{document}
\acknowledgements{We thank Tom Henzinger for an insightful discussion on Timed Logics and for encouraging us to explore non-punctual subclasses for multi-clock TPTL and ATA. We also thank Hsi-Ming Ho for an interesting discussion on handling freeze quantifiers.}
\tikzset{
  tptl/.style={shape=rectangle, rounded corners, draw, align=center,
    top color=white, bottom color=blue!20},
  level distance=2cm,
  sibling distance=7cm,
}

\maketitle

\begin{abstract}
We investigate the decidability of the ${0,\infty}$ fragment of Timed Propositional Temporal Logic (TPTL). We show that the satisfiability checking of $\tptlu$ is \pspace-complete. Moreover, even its 1-variable fragment (1-$\tptlu$) is strictly more expressive than Metric Interval Temporal Logic (MITL) for which satisfiability checking is \expspace complete. Hence, we have a strictly more expressive logic with computationally easier satisfiability checking. To the best of our knowledge, $\tptlu$ is the first multi-variable fragment of TPTL for which satisfiability checking is decidable without imposing any bounds/restrictions on the timed words (e.g. bounded variability, bounded time, etc.). The membership in \pspace is obtained by a reduction to the emptiness checking problem for a new ``non-punctual’’ subclass of Alternating Timed Automata with multiple clocks called Unilateral Very Weak Alternating Timed Automata ($\vwatau$) which we prove to be in $\pspace$.  We show this by constructing a simulation equivalent non-deterministic timed automata whose number of clocks is polynomial in the size of the given $\vwatau$.
\end{abstract}

\section{Introduction}
Metric Temporal Logic ($\mtl[\until_I, \since_I]$) and Timed Propositional Temporal Logic \\($\tptl[\until_I, \since_I]$) are natural extensions of Linear Temporal Logic ($\ltl$) for specifying real-time properties \cite{AH94}. $\mtl$ extends the $\until$ and $\since$ modality of $\ltl$ by associating a time interval with these. Intuitively, $a \until_I b$  is true at a point in the given behaviour iff event $a$ keeps  on occurring until at some future time point within relative time interval $I$, event $b$ occurs. (Similarly, $a \since_I b$ is its mirror image specifying the past behaviour.)  
On the other hand, $\tptl$ uses freeze quantifiers to store the current time stamp. A freeze quantifier \cite{AlurH92}\cite{AH94} has the form $x.\varphi$
 with freeze variable $x$ (also called a clock \cite{patricia}\cite{pandyasimoni}).
When it is evaluated at a point $i$ on a timed word, the time stamp of $i$ (say $\tau_i$) is frozen or registered in $x$, and the formula $\varphi$ is evaluated using this value for $x$. 
 Variable $x$ is used in $\varphi$ in a constraint of the form $T-x \in I$; this constraint, when evaluated at a point $j$, checks if $\tau_j -\tau_i \in I$, where $\tau_j$ is the time stamp at point $j$. Here $T$ can be seen as a special variable giving the timestamp of the present point. For example, the formula  $\varphi = \fut x.(a\wedge \fut (b \wedge T-x \in [1,2] \wedge \fut (c \wedge T-x \in [1,2])))$
  asserts that there is a point $i$ in the future where $a$ holds and in its future there is a $b$ within interval $[1,2]$ followed by a $c$ within interval $[1,2]$ from $i$. In this paper, we restrict ourselves to future time modalities only. Hence, we use the term $\mtl$ and $\tptl$ for $\mtl[\until_I]$ and $\tptl[\until]$, respectively, and MTL+Past and TPTL+Past for $\mtl[\until_I, \since_I]$ and $\tptl [\until, \since]$, respectively. We also confine ourselves to the \emph{pointwise} interpretation of these logics \cite{patricia}. 

While these logics are natural formalisms to express real-time properties, it is  unfortunate that both the logics have an undecidable satisfiability checking problem, making automated analysis of these logics difficult in general. 
 Exploring natural decidable variants of these logics has been an active area of research since their advent  \cite{AH93}\cite{raskin-thesis}\cite{icalp-raskin}\cite{Wilke}\cite{Rabinovich}\cite{rabin}\cite{rabinovichY}. One of the most celebrated such logics is the \emph{Metric Interval Temporal Logic} ($\mitl$) \cite{AFH96},  a subclass of $\mtl$ where the timing intervals are restricted to be non-punctual i.e. non-singular (intervals of the form $\langle x, y \rangle$ where $x < y$). The  satisfiability checking for $\mitl$ formulae is \expspace complete \cite{AFH96} (the result also holds for $\mitl$ + Past). 

Every formula in  $\mtl$ can be expressed in  the 1-variable fragment of  TPTL (denoted 1-$\tptl$). Moreover, 
the above-mentioned property $\varphi$ is not expressible in  $\mtl$ + Past \cite{simoni}. Hence,  1-$\tptl$ is strictly more expressive than $\mtl$ \cite{pandyasimoni, patricia}. The Logic 1-$\tptl$ can also express $\mtl$ augmented with richer counting and Pnueli modalities. Hence, $\tptl$ is a logic with high expressive power. However, decidable fragments of $\tptl$ are harder to find.
While 1-$\tptl$ has decidable satisfiability over finite timed words \cite{OWH} (albeit with non-primitive recursive complexity), it is undecidable over infinite words \cite{OuaknineW06}. There are no known fragments of multi-variable $\tptl$ which are decidable (without artificially restricting the class of timed words). In this paper, we propose one such logic, which is efficiently decidable over both finite and infinite timed words.

 We propose a fragment of $\tptl$, called $\tptl_{0,\infty}$, where, for any formula $\phi$ in negation normal form, each of its closed subformula $\kappa$ has \emph{unilateral} intervals; that is, intervals  of the form $\langle 0, u \rangle$, or of the form $\langle l, \infty )$ (where $\langle \in \{[, (\}$ and $\rangle \in \{],)\}$). 
 The main result of this paper is to show that satisfiability checking for $\tptl_{0,\infty}$ is \pspace complete. Moreover, we show that even the 1-variable fragment of this logic is strictly more expressive than $\mitl$. \pspace completeness for satisfiability checking is proved as follows: We define a sub-class of Alternating Timed Automata (ATA \cite{Ouaknine05} \cite{LasotaW08}) called \emph{Very Weak Alternating Timed Automata with Unilateral Intervals}($\vwatau$), and show that $\vwatau$ have \pspace-complete emptiness checking. 
 A  language  preserving reduction from  $\tptl_{0,\infty}$ to $\vwatau$, similar to \cite{OWH}\cite{Ouaknine05}\cite{vardiwolper},  
 completes the proof. To our knowledge, $\vwatau$ is amongst the first known fragment of multi-clock alternating timed automata ($\ata$) with efficiently decidable emptiness checking.
 Thus, we believe that  $\tptl_{0,\infty}$ and $\vwatau$ are interesting novel additions to  logics and automata for real-time behaviours.
 
 One of the key challenges in establishing the decidability of $\vwatau$ is to show that the configuration sizes  
can be bounded. In an $\ata$, a configuration can be unboundedly large owing to several conjunctive transitions, each spawning a state with a new clock valuation.  
 We  provide a framework for compressing the configuration sizes of $\vwatau$ based on simulation relations amongst states of the $\vwatau$. We then prove  that  such compression yields a simulation-equivalent transition system 
 whose  configuration sizes are bounded. This bound allows us to give a subset-like construction resulting in a simulation equivalent (hence, language equivalent) timed automata with polynomially many clocks.

The paper is organized as follows. Section \ref{sec:prelim} defines the $\tptl$ and  $\ata$, and $(0,\infty)$ fragments of these formalisms. In Section \ref{sec:main}, we prove the \pspace emptiness checking of $\tptlu$. Section \ref{sec:exp} discusses the expressiveness of $\tptl_{0,\infty}$. Section \ref{sec:discuss} concludes our work with a discussion on the implication of our work in the field of timed logics and some interesting problems that we leave open.
\section{Preliminaries}
\label{sec:prelim}
Let $\mathbb{Z}, \mathbb{Z}_{\ge 0}, \mathbb{N}, \R, \R_{\geq 0}$ respectively denote 
 the set of integers, non-negative integers, natural numbers (excluding 0),  real numbers, and  non-negative real numbers.   Given a sequence $\mathbf{a} = a_1 a_2  \ldots$, $\mathbf{a}[i] = a_i$ denotes the $i^{th}$ element of the sequence, $a[i..j]$ represents $ a_i a_{i+1} \ldots a_j$,  $a[i..]$ represents $a_i a_{i+1} \ldots $  and $a[..i]$ represents $ = a_1 a_2 \ldots a_i$. Let $\intinterval$ be the set of all the open, half-open, or closed intervals (i.e.\ convex subsets of real numbers), such that the endpoints of these intervals are in $\mathbb{N} \cup \{0,\infty\}$. Intervals of the form $[x,x]$  are called  punctual; a non-punctual interval is one which is not punctual. For, $\langle {\in} \{(,[\}$ and $\rangle {\in} \{ ],) \}$, 
an interval of the form $\langle 0, u \rangle$ for $u >0$ 
is called \emph{right-sided} while an interval of the form  $\langle l, \infty)$ is called \emph{left-sided}. A \emph{unilateral} interval is either left-sided or right-sided. 
Let $\intintervalz, \intintervalinf \subseteq \intinterval$ respectively be the set of all \emph{right sided} and 
\emph{left sided} intervals of the form $\langle 0, u \rangle$, $\langle l, \infty)$, for any $l, u \in \mathbb{Z}_{\ge 0}$. Let $\intintervalu = \intintervalz \cup \intintervalinf$.
For $\tau {\in} \R$ and interval  $\langle a, b\rangle$, ${\tau+\langle a, b\rangle}$ stands for the interval ${\langle \tau+a, \tau+b \rangle}$. 
\\\textbf{Timed Words.} Let $\Sigma$ be a finite alphabet. 
 A finite (infinite) word over $\Sigma$ is a finite (infinite) sequence over  $\Sigma$. The set of all the finite (infinite) words over $\Sigma$ is denoted by $\Sigma^*$ ($\Sigma^\omega$). A finite timed word $\rhos$ over $\Sigma$ 
  is a finite sequence of pairs $(\sigma, \tau)  \in (\Sigma \times \R_{\geq 0})^*$ : $\rhos = (\sigma_1, \tau_1), \ldots, (\sigma_n, \tau_n)$ where $\tau_i \leq \tau_j$ for all $1 \leq i \leq j \leq n$. Let $dom(\rhos) = \{1,2, \ldots n\}$ be the set of points in  $\rhos$. Likewise, 
an infinite timed word is an infinite sequence  $\rhos = (\sigma_1, \tau_1) (\sigma_2, \tau_2) \ldots \in  (\Sigma \times \R_{\geq 0})^{\omega}$, 
 where $\sigma_1 \sigma_2 \ldots \in \Sigma^\omega$, and $\tau_1 \tau_2 \ldots $ is a monotonically increasing infinite sequence of real numbers approaching $\infty$ (i.e. non-zeno).  
A finite (infinite) timed language is a set of all finite (infinite) timed words over $\Sigma$ denoted $T\Sigma^*$  ($T\Sigma^{\omega}$). 
\\{\bf{Timed Propositional Temporal Logic (TPTL)}}.  
The logic $\tptl$ extends $\ltl$ with freeze quantifiers and is evaluated on timed words. 
Formulae of $\tptl$ are built from a finite alphabet $\Sigma$  using Boolean connectives, 
as well as the temporal modalities  of $\ltl$. In addition, $\tptl$ uses a finite set of real-valued variables called  freeze variables or \emph{clocks}  $X = \{x_1,\ldots,x_n\}$. Let $\nu: X \rightarrow \mathbb{R}_{\geq 0}$ represent  a valuation 
assigning a non-negative real value to each clock. Without loss of generality, we work with $\tptl$ in the negation normal form, where all the negations appear only with  atomic formulae.  Formulae of $\tptl$ are defined as follows.

$\varphi::=a~|~\neg a ~|\top~|~\bot~|~ x.\varphi ~|~ T-x \in I ~|~\varphi \wedge \varphi~|~ \varphi \vee \varphi~|
~\varphi \until \varphi ~|~ \sbf \varphi ,$
\\where $x \in X$, $a \in \Sigma$, $I \in \intintervaln$\footnote{Notice that duals of Until; ``Unless'' and ``Release'' operators can be expressed using a $\sbf$ and an $\until$ operator without compromising on succinctness.}. $T$ denotes the time stamp of the position where the formula is  evaluated. 
The construct $x. \varphi$ is called a \emph{freeze quantifier}, which stores in $x$,  the time stamp of the current position  and then evaluates $\varphi$.  $T-x \in I$ is a constraint on the clock variable $x$, which checks if the time elapsed since the time 
$x$ was frozen is in the interval $I$. \textbf{Notice} that, in aid of brevity, we will typically \textbf{abbreviate} subformula $T-x \in I$ to $x \in I$. 
For a timed word $\rho=(\sigma_1,\tau_1)\dots(\sigma_n,\tau_n)$, $i \in dom(\rho)$ and a $\tptl$ formula $\varphi$, we define the satisfiability  $\rho, i, \nu \models \varphi$ at a position 
$i$ of $\rho$, given a valuation $\nu$ of  the clock variables. 
\begin{flalign*}
\rho, i, \nu &\;\;\;\models & a\;\;\;\;\;\;\;\;\;\;\;\;\;\;\;\;\;\;\;\;\;\;\;\;\;\;\;\;&  \iff & \sigma_{i}=a,\Repeat{62}{\gap}\\ 
\rho,i,\nu & \;\;\;\models & x.\varphi\;\;\;\;\;\;\;\;\;\;\;\;\;\;\;\;\;\;\;\;\;\;\;\;\;& \iff & \rho,i,\nu[x \leftarrow \tau_i] \models \varphi,\Repeat{43}{\gap}\\
\rho,i,\nu & \;\;\;\models & {T-x \in I} \;\;\;\;\;\;\;\;\;\;\;\;\;\;\;&  \iff & \tau_i - \nu(x) \in I,\Repeat{51}{\gap} \\
\rho,i,\nu& \;\;\;\models & \sbf \varphi \;\;\;\;\;\;\;\;\;\;\;\;\;\;\;\;\;\;\;\;\;\;\;\;\;\;& \iff & \forall j > i, \rho,j,\nu \ \models\ \varphi,\Repeat{42}{\gap}\\
\rho,i, \nu& \;\;\;\models & \varphi_{1} \until \varphi_{2} \;\;\;\;\;\;\;\;\;\;\;\;\;\;\;\;\;\;\;\; & \iff & \exists j > i,
\rho,j\ \models  \varphi_{2}, \text{ and }  \forall i< k <j, \rho,k\ \models\ \varphi_{1}.
\end{flalign*}

The $\fut$ and $\nx$ operator is  defined in terms of $\until$; $\fut \phi = \top \until \phi$ and $\nx\phi = \bot \until \phi$.  $\mathbf{0}$ denotes a valuation that maps every variable to $0$. 
A TPTL formula $\varphi$ is said to be closed iff every variable $x$ used in the timing constraint is quantified (or bound) by a freeze quantifier. A formula that is not closed is \emph{open}.  Similarly, in any formula $\varphi$, a constraint of the form $x \in I$ is open if $x$ is not quantified.
For example, $x.y.(a \until (b \wedge x \in (1,2) \wedge y\in (2,3)))$ is a closed formula while $x.(a \wedge \underline{y\in (2,3)})\until y.(b \wedge x \in (1,2))$ is open as the clock $y$ used in the underlined clock constraint is not in the scope of a freeze quantifier for $y$. Moreover, the underlined constraint $\underline{y\in (2,3)})$ is an open constraint. Notice that open constraints appear only (and necessarily) in open formulae.
Satisfaction of closed formulae is independent of the clock valuation; that is, if $\psi$ is a closed formula, then for a timed word $\rho$ and a position $i$ in $\rho$, either for every valuation $\nu$, $\rho, i, \nu \models \psi$; or for every valuation $\nu$, $\rho, i, \nu \not\models \psi$. Hence, for a closed formula $\psi$, we drop the valuation $\nu$ while evaluating satisfaction, and simply write $\rho,i\ \models \psi$. As an example, the closed formula $\varphi{=}x.(a \until (b \until (c \wedge x {\in} [1,2])))$ 
is satisfied by the timed word $\rho{=}(a,0)(a,0.2)(b,1.1)(b,1.9)(c,1.91)(c,2.1)$ since $\rho, 1 \models \varphi$.  
The word $\rho'=(a,0)(a,0.3)(b,1.4)(c,2.1)(c,2.5)$ does not satisfy $\varphi$. 
However, $\rho',2 \models \varphi$: if we start from the second position of $\rho'$, 
the value 0.3 is stored in $x$ by the freeze quantifier, and  when we reach the position 4 of $\rho'$ with $\tau_4=2.1$
we obtain $T-x=2.1-0.3 \in [1,2]$.

Given any closed TPTL formula $\varphi$, its language, $L(\varphi) = \{\rho | \rho, 1 \models \varphi\}$, is set of all the timed words satisfying it. We say that a closed formula $\varphi$ is satisfiable iff $L(\varphi)\neq \emptyset$. 
\\{\textbf{Size of a TPTL formula}}.
Given a TPTL formula $\varphi$, the size of $\varphi$ denoted by $|\varphi|$ is defined as 
$B+M+C$ where $B$ is the number of Boolean operators in $\varphi$, $M$ is the  number of temporal modalities
($\sbf, \until, \nex, \fut$) and freeze quantifiers in $\varphi$, and $C$ is obtained by multiplying the number of time constraints in $\varphi$ with 
 $2 \times (\lfloor \log(c_{max})\rfloor + 1)$ where $c_{max}$ is the maximal constant appearing in the time constraints 
 of $\varphi$.  
 For example, for $\varphi = x.(a \wedge b U (c \vee x \le (1,2)))$, $|\varphi| = 2 + 2 + 2\times (1+1) = 8$ as it contains two boolean operators, one temporal modality, one freeze quantifier and one timing constraint where $c_{max} = 2$.

The subclass of $\tptl$ that uses {\bf only k-clock variables} is known as \textbf{k-TPTL}. By \cite{OWH} \cite{OuaknineW06}, satisfiability checking for 1-$\tptl$ is decidable over finite models but non-primitive recursive hard, and undecidable over infinite models. 
Satisfiability checking for 2-$\tptl$ is undecidable over both finite and infinite models \cite{AH94} \cite{ictac}. 
Towards the main contribution of this paper, we propose a `non-punctual' fragment of $\tptl$ with unilateral intervals, called $\tptlu$, and 
show that its satisfiability checking is decidable with multiple variables over both finite and infinite timed words (\pspace-complete). Further, 1-$\tptlu$ is already more expressive 
than $\mitl$, which has an \expspace-complete satisfiability checking. 

\subsection{Multi-clock TPTL with unilateral intervals : $\tptlu$}
\label{sec:tptlu}
We say that a formula $\varphi$ is of the type $\le$ ($\ge$), iff all the intervals appearing in the open constraints of $\varphi$ are in $\intintervalz$ ($\intintervalinf$). Notice that a closed formula belongs to both types $\le$ and $\ge$. There are open formulae that are neither of type $\le$ nor $\ge$.
A $\tptl$ formula $\varphi$ in negation normal form is a $\tptlu$ formula iff every subformula of $\varphi$ is either of the type $\le$ or $\ge$.
For example, $x.y.(a \until (b \until (c \wedge x < 3 \wedge y \le 2 \wedge x.(\nex(c \wedge x>1)))))$ is a $\tptlu$ formula since there is no subformula that doesn't belong to either types $\le$ or $\ge$.
However, $x.y.(a \until (b \wedge x \le 3 \wedge y \ge 5))$ is not $\tptlu$, since 
$(b \wedge x \le 3 \wedge y \ge 5)$ is of neither type $\le$ or $\ge$ as the open constraints within this subformula use both left-sided as well as right-sided intervals. 
This restriction is inspired by that of $\mitl^{0,\infty}$. Any $\mitlu$ formula can be expressed in $1$-$\tptlu$  by applying the same reduction from $\mitl$ to $1$-$\tptl$ (see Remark \ref{rem:mitl-tptl}). Next, we introduce alternating timed automata which are useful in proving the main result, i.e., Theorem \ref{thm:main}.

\subsubsection{Alternating Timed Automata}
\label{sec:ata}
An Alternating Timed Automata (ATA) 
is a 7-tuple $\Aa = (Q, \Sigma, \delta, q_0, \Qacc, X, \G)$, where, $Q$ is a finite set of locations, $X$ is a finite set of clock variables, $\G$ is a finite set of guards of the form $x \in I$ where $I \in \intintervaln$ and $x \in X$, $\delta$ is a transition function, $q_0 \in Q$ is the initial location, and $\Qacc \subseteq Q$ is a set of accepting locations. The transition function is defined as   
$\delta: Q \times \Sigma \mapsto \Phi(Q,\G)$ where $\Phi(Q,\G)$ is defined by the  grammar
$\varphi::=\top|\bot|\varphi_1 \wedge \varphi_2|\varphi_1 \vee \varphi_2|q|x \in I|Y.q$
with  $q \in Q$, $x \in X$, $(x \in I)$ is a guard in  $\G$, 
$Y \subseteq X$, $Y$ is not the empty set. $\top, \bot$ respectively denote  $\true$ and  $\false$. 
 $Y.q$ is a \emph{binding construct} which resets all clocks in $Y$ to zero after taking the transition. 
 \begin{figure}
  \begin{center}
    \includegraphics[scale=0.27]{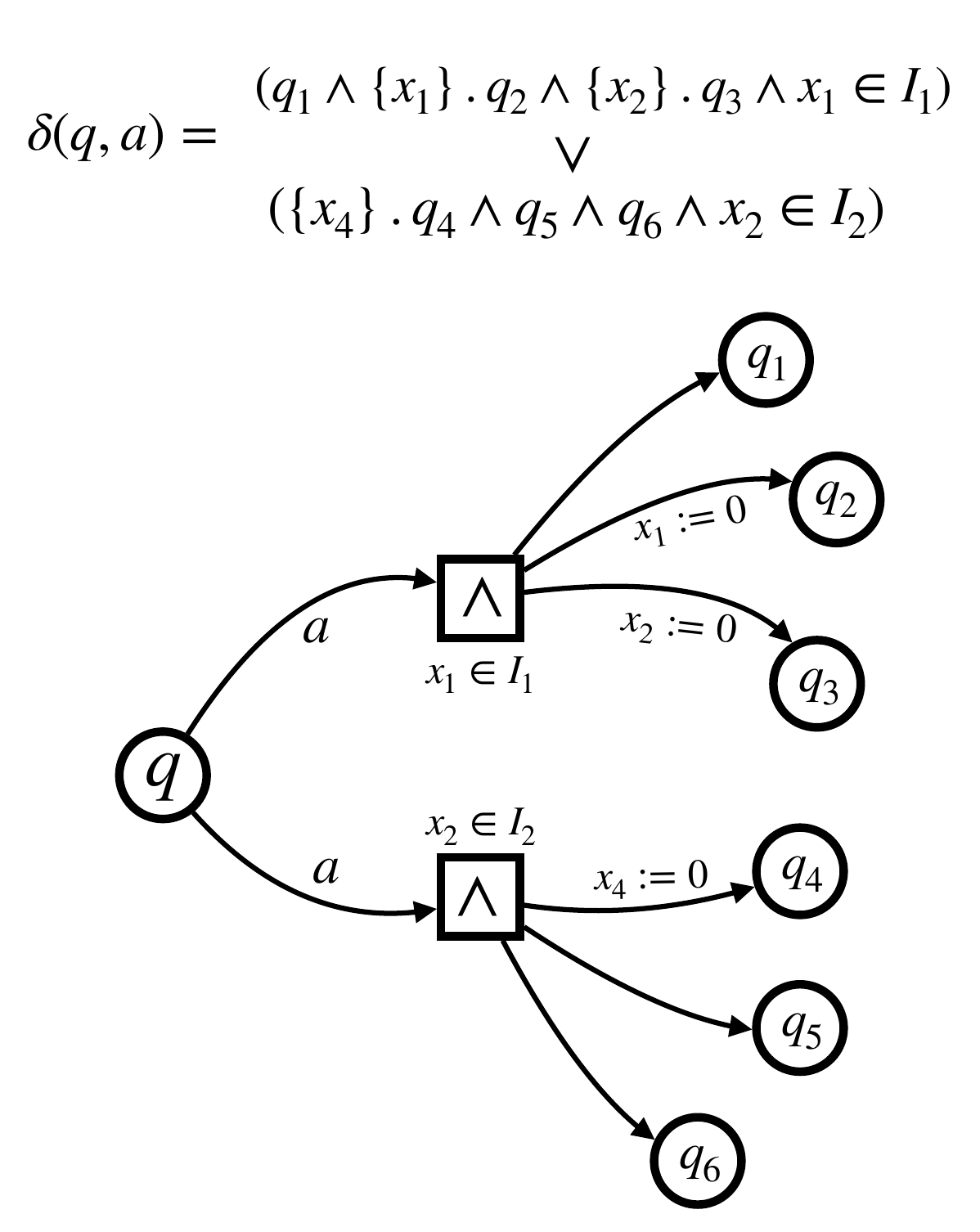}
    \caption{Graphical Representation of ATA transitions.}
    \label{fig:ATA}
    \end{center}
\end{figure}
 Let $p,q \in Q$ and $Y \subseteq X$. We say that there is a transition from $p$ to $q$ iff $q$ appears in $\delta(p,b)$ for some $b \in \Sigma$.
We say that there is a \textbf{strong reset transition}, \textbf{non-reset} transition, and a \textbf{Y-reset} transition from location $p$ to $q$ iff for some $b \in \Sigma$, $X.q$, $q$, and $Y.q$, respectively, appears in $\delta(p,b)$ for some $b \in \Sigma$. 
The 1-clock restriction of ATA has been considered in \cite{Ouaknine05} and \cite{LasotaW08}.  
\\{\bf Evaluation of $\Phi(Q,\G)$}. Given an ATA $\Aa$,  a state $s$ is defined as a pair consisting of a  location and a valuation over $X$, i.e., $s \in Q\times \Nu_X$. A configuration $C$ of an ATA is a finite set of states. Let $S$ and $\mathcal{C}$ respectively denote the set of all states and configurations of $\A$. A configuration $C$ and a clock valuation $\nu$ define a Boolean valuation for $\Phi(Q,\G)$ as follows:\\
\begin{tabular}{l|l}
$C \models_{\nu} q$ iff $(q,\nu)\in C$, & $C \models_\nu Y.q$ iff $(q,\nu) \in C$, and $\forall x \in Y. \nu(x) = 0,$\\  
$C \models_{\nu} x\in I$ iff $\nu(x) \in I$, & $C \models_{\nu} \varphi_1 \wedge \varphi_2$ iff $C\models_{\nu} \varphi_1 \wedge C \models_{\nu}\varphi_2,$\\
$C \models_{\nu} \top$ for all $C\in \mathcal{C}$, &  $C \models_{\nu} \varphi_1 \vee \varphi_2$ iff  $C \models_{\nu} \varphi_1 \vee C \models_{\nu}\varphi_2.$\\ 
\end{tabular}
\\Finally, $C \not \models_{\nu} \bot$ for all possible configurations. We say that $C$ is a minimal model for $\varphi \in \Phi(Q, \G)$ with respect to $\nu$ (denoted by $C \modelmin_\nu \varphi$) iff $C \models_{\nu} \varphi$ and no proper subset $C'$ of $C$ is such that $C' \models_{\nu} \varphi$.
We represent the transition function of ATA (and hence by extension the ATA) graphically as shown in figure \ref{fig:ATA}.

\noindent{\bf Semantics of  ATA}.  Given a state $s=(q,\nu)$, a time delay $t \in \R_{\ge 0}$ and $a \in \Sigma$, the successors  of $s=(q,\nu)$ on time delay $t$ followed by $a$ is any configuration $C$ such that $C \modelmin_{\nu +t} \delta (q, a)$. $\suca^{st}(s,t,a)$ is the set of all such successors. 
The notion of a successor is extended to a configuration in a straightforward manner. 
A configuration $C'$ is a successor of configuration $C = \{s_1, s_2, \ldots s_k\}$ on time delay $t$ and $a \in \Sigma$ (denoted by $C \xrightarrow{(t,a)}_\Aa C'$) iff $C' = C_1 \cup \ldots \cup C_k$ such that $\forall 1\le i \le k, C_i\in \suca^{st}(s_i,t,a)$. We denote by $\succ(C,t,a)$ set of all such successors $C'$.

 The initial configuration is defined by $C_{init} = \{(q_0, \mathbf{0})\}$, and 
a configuration $C$ is accepting iff for all $s \in C$, $s$ is an accepting state, that is $s=(q, \nu)$ for $q \in \Qacc$. Let \Cac be the set of all the accepting configurations. 
Hence, the empty configuration is an accepting configuration. We define the semantics of ATA using a Labelled Transition System (LTS). An LTS is a 5-tuple $T = (S, s_0, \Sigma, \delta, S_f)$, where $S$ is a finite or infinite set of states, $s_0 \in S$ is the initial state, $\Sigma$ is set of symbols, $\delta: S \times \Sigma \times S$ is a transition relation, and $S_f \subseteq S$ is a set of final states. A (finite) run $\rs$ of an LTS  is a (finite) sequence of the form $s_0, a_1, s_1, a_2, s_2, a_3 \ldots $ where $s_1, s_2, \ldots \in S$ are states of $T$, and $a_1, a_2, \ldots $ are symbols in $\Sigma$ such that for all $i > 0$, $s_i \in \delta(s_{i-1}, a_i)$. We say that a run $R = s_0, a_1, s_1, a_2, s_2, a_3 \ldots $ visits a state $s$ (or visits a set of states $S'$) iff the sequence $R$ contains $s$ (or contains states in $S'$). A  run is said to be accepting iff it ends in some state $s \in S_f$. Similarly, an infinite run is said to be \buchi accepting iff it visits $S_f$ infinitely often. 

Runs of $\Aa = (Q, \Sigma, \delta, q_0, \Qacc, X, \G)$ starting from a configuration $C$ are the runs of LTS $TS(\Aa,C)=(\mathcal{C}, C,  \mathbb{R}_{\ge 0} \times \Sigma, \rightarrow,$ \Cac$)$. Notice that the states of LTS $TS(\Aa, C)$ are configurations of $\Aa$ (i.e., a set of states of $\Aa$ and not just the states of $\Aa$). Let  $\rhos = (a_1, \tau_1) (a_2, \tau_2) \ldots$ be any timed word over $\Sigma$. We say that a run $R = C, (t_1, a_1), C_1, (t_2, a_2)\ldots$ is produced by $\Aa$ on $\rhos$ starting from a configuration $C$ iff $C \xrightarrow{(t_1, a_1)} C_1 \xrightarrow{(t_2, a_2)}C_2 \ldots$ where $t_i = \tau_i-\tau_{i-1}$ for $i>0$ and $\tau_0 = 0$. Let $\Aa(\rhos, C)$ be the set of all the runs produced by $\Aa$ on $\rhos$, starting from the configuration $C$. We denote $TS(\Aa, C_{init})$ as simply $TS(\Aa)$.
A run starting from the initial configuration $C_{init}$ is called an initialized run. We denote  $\Aa(\rhos, C_{init})$ by  $\Aa(\rhos)$. 
 $\rhos, i$ is said to be accepted (\buchi accepted) by $\Aa$ starting with configuration $C$, denoted by $\rhos,i \models \Aa,C$, iff there exists a run in $\Aa(\rhos[i..], C)$ accepted (\buchi accepted) by $TS(\Aa)$ (i.e., simulating $\Aa$ on the suffix of $\rho$ starting at position $i$ we obtain an accepting run). We say that $\rhos$ is accepted by $\Aa$ iff $\rhos, 1 \models C_{init}$.

 We define the finite (infinite) language of $\Aa$, denoted by $L_{fin}(\Aa)$ ($L_{inf}(\Aa)$), as a set of all the finite (infinite) timed words accepted by $\Aa$. When clear from context, we drop the subscript in $L_{fin}$ and $L_{inf}$.

Non-Deterministic Timed Automata (NTA) is a subclass of ATA where $\Phi(Q,\G)$ is restricted to be in disjunctive normal form (DNF), where each disjunct is of the form  $(q \wedge x \in I)$ or $(X'.q \wedge x \in I)$. Hence, for any $s \in S$, $t\ge 0$, $a \in \Sigma$ and any configuration $C \in \succ^{st}(s,t,a)$ implies  $C\le 1$. 

We call the ATA $\Aa$ a \textbf{Very Weak ATA} (VWATA) iff (1) there is a partial order $\ll_{\Aa} \subseteq Q \times Q$ such that there is a transition from $p$ to $q$ iff $q \ll_{\Aa} p$, (2) all the self-loop transitions (transitions entering and exiting into the same location) are non-reset transitions, and (3) For every location $q$, there is at most one location $p\ne q$ such that there is a transition from $p$ to $q$. Moreover, all the transitions from $p$ to $q$ reset the same set of clocks. This makes the transition diagram of $\vwata$ a tree and not a DAG (excluding self-loops). 
\begin{remark}
\label{sec:vwatanew-vwata}
In the literature, VWATA (also called Partially-Ordered Alternating Timed Automata in \cite{KKP18}) and their corresponding untimed version \cite{gastin-oddoux}\cite{vw}(also called as Linear \cite{linear}, Linear-Weak \cite{linearw}, 1-Weak \cite{1w}, and Self-Loop \cite{slaa} Alternating Automata) are required to satisfy only conditions (1) and (2). It can be shown that condition (3) does not affect the expressiveness of the machine. We notice that this version of $\vwata$ is enough to express TPTL formulae efficiently (linear in the size of TPTL formulae). In case of translation from $\tptl$ to VWATA satisfying condition (3) the number of locations in the resulting ATA will depend on the size of the formula tree. On the other hand, the total number of locations depends on the formula DAG on similar translation from $\tptl$ to VWATA satisfying only (1) and (2) making it exponentially more succinct. Hence, we consider a less succinct representation (i.e., tree or string, which is standard) of $\tptl$ formulae for computing its size as compared to the DAG representation.
\end{remark}
\subsubsection{ATA with Unilateral Intervals: $\atau$}
\label{sec:atau}
Similar to the unilateral version of $\tptl$ (i.e. $\tptlu$), we define a unilateral version of $\ata$ as follows.
Let $\Aa = (Q, \Sigma, \delta, q_0, \Qacc, X, \G)$ be any ATA. Let $\Gg$ ($\Gl$) be the subset of $\G$ containing all the guards of the form $x \in I$ where $I \in \intintervalinf$ ($I \in \intintervalz$). $\Aa$ is said to be an $\atau$ iff, $Q$ can be partitioned into $\Qg$ and $\Ql$ any transition exiting from any location $q \in \Qg$ ($q \in \Ql$) is guarded by a guard in $\Gg$ ($\Gl$), and any transition from any location in $\Qg$ to a location in $\Ql$, or vice-versa, is a strong reset transition. $\Aa$ is said to be $\vwatau$ iff it is also a VWATA. From this point onwards, for any set of locations $Q$ of any $\atau$, $\Qg$ and $\Ql$ will denote partitions of $Q$ satisfying the above condition.


\section{Satisfiability Checking for $\tptlu$}
\label{sec:main}
This section is dedicated to proving the following main theorem  of this paper.
\begin{theorem}
\label{thm:main}
\end{theorem}
\pspace  hardness follows from the hardness of satisfiability checking of the sublogics $\ltl$ and $\mitlu$ (see section \ref{sec:mE} for the details on $\mitlu$).
To show membership in \pspace we propose the following steps: (1) We reduce any given $k$-$\tptlu$ formula $\varphi$, to an equivalent $\vwatau$, $\Aa$, with $k$ clock variables and at most $|\varphi|+1$  number of locations. (2) We give a novel on-the-fly construction from any $\vwatau$ to simulation equivalent NTA $\A$ with exponential blow-up in the number of locations and polynomial blow-up in the number of clocks. Hence, the region automata corresponding to $\A$ has at most exponentially many states, and thus each state can be represented in polynomial space. 
\begin{remark}
\label{rem:complpspace}
Notice that while the reduction from $\vwatau$ to timed automata results in an exponential blow-up in the number of locations we can directly construct the region automaton of the corresponding timed automaton on-the-fly making sure that we need at most polynomial space to solve its emptiness checking problem.
\end{remark}
We demonstrate our steps of construction using a running example and defer the formal constructions to the Appendices \ref{app:tptltoata}, \ref{app:atatonta}. 
In our running example, we start with the given formula $\varphi = \sbf(\neg a \vee x.(\fut(a \wedge T-x \le 2 \wedge y. \nex(b \wedge T-x \le 3 \wedge T-y \le 2)))$. 
\subsection{$\tptlu$ to $\vwatau$}
\label{sec:tptl-vwata}
\begin{figure}
    \begin{center}
 
    \includegraphics[scale=0.25]{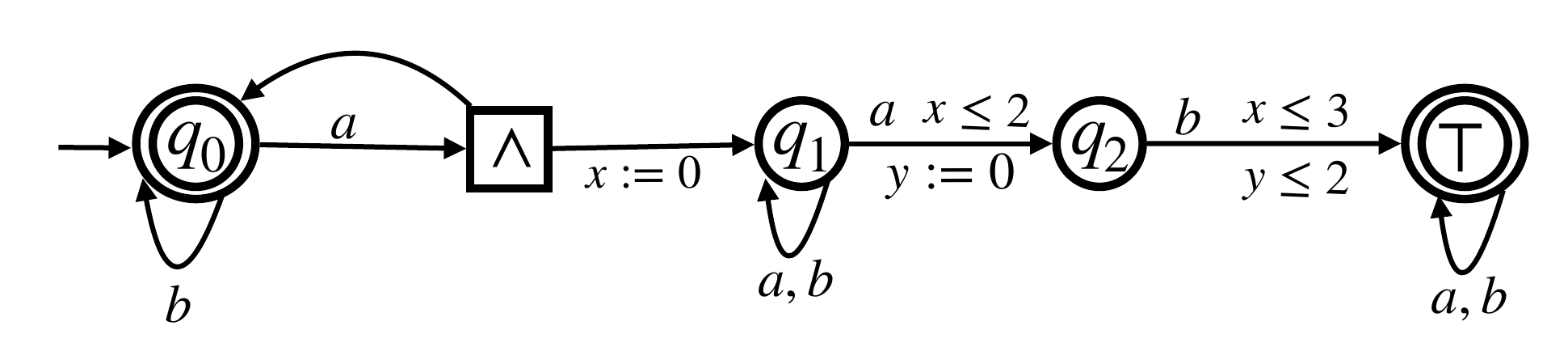}
    \caption{  $\vwatau$ equivalent to  $\varphi$. Location $q_i$ corresponds to the subformula $\varphi_i$ :  $\rhos, i, \nu \models \varphi_i$ iff $\rhos, i \models (q_i, \nu)$.}
    \label{fig:running-example}
    \end{center}
\end{figure}
This step is a straightforward multi-clock generalization of translation from $\mtl$ and $1$-$\tptl$ to 1-ATA in \cite{Ouaknine05} and \cite{OWH}, respectively, (which are themselves timed generalization of reduction from $\ltl$ to Very Weak Alternating Automata \cite{vardiwolper} \cite{gastin-oddoux}). We give the reduction in Appendix \ref{app:tptltoata} for completeness. The proof of equivalence is identical to that in \cite{Ouaknine05} and \cite{OWH} resulting in the following Theorem \ref{thm:tptltoata}. We give the $\vwatau$ corresponding to the formula $\varphi$ of the running example in Figure \ref{fig:running-example}. Hence, to prove the main theorem it suffices to show that emptiness checking for $\vwatau$ is in \pspace (i.e. Theorem \ref{thm:vwatasat}).
\begin{theorem}
\label{thm:tptltoata}
    Any $k$ variable $\tptl$ formula $\varphi$ over $\Sigma$ can be reduced to an equivalent $\vwata$, $\Aa = (Q, 2^\Sigma, \delta, \init, \Qacc, X, \G)$, with $|X| = k$, $|Q| \le |\varphi| + 1$, and  $\G$ is the set of all the guards appearing in $\varphi$. Moreover, if $\varphi$ is a $\tptlu$ formula, then the $\Aa$ is $\vwatau$.
\end{theorem}
\subsection{Emptiness Checking for $\vwatau$}
\label{sec:vwatauemptiness}
 The following theorem is the main technical result.
\begin{theorem}
\label{thm:vwatasat}
    Emptiness Checking for $\vwatau$ is in \pspace.
\end{theorem}
We give a translation from $\vwatau$ {$\Aa = (Q, \Sigma, \delta, q_0, \Qacc, X, \G)$} to an equivalent timed automaton, ${\A = (\Q, \Sigma, \Delta, \q_0, \Qc, \X, \G)}$, such that the transition system of $\Aa$ (i.e., $TS(\Aa)$) is simulation equivalent to that of $\A$ (i.e., $TS(\A)$). Hence, by the Proposition \ref{prop: langeq}, ${L(\Aa) = L(\A)}$.

Moreover, $\Q = O(2^{Poly(Q)})$ and $|\X| = |X| \times |Q|$. Hence, the number of states in the corresponding region automaton is exponential to the size of $\Aa$ (i.e. $O(2^{Poly(|Q|, |X|)}) \times (2 \times c_{max}+1)$ where $c_{max}$ is the maximum constant used in the constraints appearing in $\G$). Hence, each state of the region automata (when encoded in binary) can be represented in polynomial space proving membership in \pspace.
We prove the above by giving a translation from $\vwatau$ to timed automata with polynomial blowup in the number of clocks and exponential blowup in the set of locations. 
As a side-effect, we also show that emptiness checking for 1-$\atau$ is in \pspace (using the same construction) generalizing the result of \cite{H19}. We first briefly discuss the concept of simulation relations and preorder.

 \subsubsection{Simulation Relations and Preorder}
 We fix a pair of labeled transition system, $TS^1 = (S^1, s^1_0, \Sigma, \delta^1, S_{f}^{1})$ and $TS^2 = (S^2, s^2_0, \Sigma, \delta^2, S^{2}_{f})$. A relation ${\preceq} \subseteq S^1 \times S^2$ is a simulation relation iff 
 (1) $s^1_0 {\preceq} s^2_0$,
 (2) for every $s_1 \preceq s_2$, (2.1) if $s_1 \in S^1_{f}$ then $s_2 \in S^2_{f}$, and (2.2) for every $a \in \Sigma$, for every $s_1' \in \delta(s_1, a)$ there exists $s_2' \in \delta(s_2, a)$ such that $s_1' \preceq s_2'$. If $s_1 \preceq s_2$, then we say that $s_2$ simulates $s_1$ wrt $\preceq$.

 Let $S = S^1 \cup S^2$. Notice that simulation relations are closed under union. Hence, there is a unique maximal simulation relation, ${\le} \subseteq S \times S$, which is the union of all the simulation relations amongst states of $TS^1$ and $TS^2$ (i.e. all the simulation relations between $TS^1$ and itself, between $TS^2$ and itself, and from $TS^1$ to $TS^2$ and vice-versa). Notice that $\le$ is a preorder relation (i.e. reflexive and transitive), and hence also called simulation preorder. Similarly, simulation equivalence relation, $\cong$ is defined as the largest symmetric subset of simulation preorder, $\le$. I.e., $s \cong s'$ iff $s \le s'$ and $s' \le s$. Hence, it is clear that $\cong$ is an equivalence relation. If $s \le s'$ we say that $s'$ simulates $s$. Recall that the states of $TS(\Aa,C)$ for any ATA $\Aa$ and its configuration $C$ are configurations of $\Aa$. The following Proposition is then straightforward. 
\begin{proposition}
\label{prop: langeq}
    Let $\Aa$ and $\Aa'$ be any ATA, and $s_0, s'_0$ be their initial states, respectively. $TS(\Aa,\{s\}) \le TS(\Aa', \{s'\})$ implies $L_{fin}(\Aa) \subseteq L_{fin}(\Aa')$ and $L_{inf}(\Aa) \subseteq L_{inf}(\Aa')$. Hence,  $TS(\Aa,\{s\}) \cong TS(\Aa', \{s'\})$ implies $L_{fin}(\Aa) = L_{fin}(\Aa')$ and $L_{inf}(\Aa) = L_{inf}(\Aa')$
\end{proposition}
We fix an ATA $\Aa = (Q, \Sigma, \delta, q_0, \Qacc, X, \G)$. Let $C$ and $C'$ be arbitrary configurations of $\Aa$. 
Let $\la$, $\ca$ be the simulation preorder and simulation equivalence amongst configurations of $\Aa$. That is, $C \la C'$ iff  $C'$ simulates $C$, and $C \ca C'$ iff $C$ is simulation equivalent to $C'$ in $TS(A)$, the transition system corresponding to ATA $\Aa$. Then, by Proposition \ref{prop: langeq}:
\begin{remark}
    \label{rem:langsim}
    For any configuration $C$ and $C'$ of $A$, $C \la C'$ implies $L(\Aa, C) \subseteq L(\Aa, C')$ and $C \ca C'$ implies $L(\Aa, C) = L(\Aa, C')$.
\end{remark}

\begin{remark}    
        \label{rem:subsetsimulation}
        $C \supseteq C'$ implies $C \la C'$. Hence, for any timed word $\rho$, if $\rho, i \models \Aa, C$ then $\rho, i \models \Aa, C'$. Intuitively, the additional states in $C$ (which are not appearing in $C'$) impose extra obligations in addition to that imposed by states common in both $C$ and $C'$ which makes reaching the accepting configuration (hence accepting a timed word) harder from $C$.
    \end{remark}
\begin{proof}[Proof of Remark 8]
    We show that $\supseteq$ is a simulation relation. That is, $C’ \subseteq C$ (or $C \supseteq C’$) implies $C \la C’$. We just need to prove that $\supseteq$ indeed satisfies condition 2 (condition 1 is trivially satisfied) of the simulation relation. Notice that the set of accepting configurations is downward closed.
That is, for any accepting configuration $D$ any of its subset $D’$ is accepting (by definition of accepting configuration). Hence,  $\supseteq$ satisfies condition 2.1.

Wlog, let $C = \{s_1, \ldots s_n\}$ and $C’  = \{s_1, \ldots s_m\}$ for some $m \le n$.
For any $(t,a) \in \R \times \Sigma$, by definition, any $C \xrightarrow{(t,a)} D$ iff $D = D_1 \cup D_2 \ldots D_n$ where $D_i$ is some successor of $s_i$ on $(t,a)$. Similarly, $C' \xrightarrow{(t,a)} D' iff D’ = D_1 \cup D_2 \ldots D_m$ where $D_i$ is some successor of $s_i$ on $(t,a)$. Hence, For any $(t,a) \in \R \times \Sigma$, for any $C \xrightarrow{(t,a)} D$ there exists a $C' \xrightarrow{(t,a)} D'$ such that $D \supseteq D’$ (for any state $s_i \in C \cap C’$ , just choose the same successor of $s_i$ which was used in $D$ to construct $D’$). Hence, for any successor $C$, we have a successor of $C’$ which is a subset of that of $C$. Hence, $\supseteq$ satisfies condition 2.2.

Intuitively, these extra states $s_{m+1}\ldots s_n$ will generate extra states in the successor configurations, which will make reaching an accepting configuration harder from $C$ (and its successors) as compared to $C’$ (and its successors),
\end{proof}
    
\begin{remark}
    \label{rem:replacesimulation}
        If $D' \subseteq C$ and $D \la D'$,  then $(C\setminus D') \cup D \la C$. In other words, we can replace the states in $D'$ with that in $D$ in any configuration $C$, and get a configuration that is simulated by $C$. Hence, $L(\Aa, (C\setminus D') \cup D) \subseteq L(\Aa, C)$.
    \end{remark}
\begin{proof} [Proof of Remark 9]
Consider a relation $\r$ amongst configurations of $A$ such that $E~{\r}~E'$ iff $E' = E_1 \cup E_2$ such that $E \la E_1$ and $E \la E_2$. We first show that $\r$ is a simulation relation. We need to show that $\r$ satisfies condition (2.2) as it trivially satisfies condition (2.1). Let $E~\r~E'$ and $E' = E_1 \cup E_2$ such that $E \la E_1$ and $E \la E_2$.  Let $E_2' = E_2 \setminus E_1$. Then $E \la E_2 \la E_2'$ (by remark \ref{rem:subsetsimulation}). By definition of simulation preorder, for any $E \xrightarrow{(t,a)} F$ there exists (1)$ E_1  \xrightarrow{(t,a)} F_1$ and (2)$E_2' \xrightarrow{(t,a)} F_2$ such that $F \la F_1$ and $F \la F_2$. Let $F' = F_1 \cup F_2$. By semantics ATA, (1) and (2) imply $E' \xrightarrow{(t,a)} F'$.   
Moreover, by definition of $\r$, $F~\r~F'$ if $E~\r~E'$. Hence, for any $E~\r~E'$, and $E \xrightarrow{(t,a)} F$ there exists $E' \xrightarrow{((t,a)} F'$ such that $F~\r~F'$. Hence, $\r$ satisfies condition 2.2.

Given $D' \la D$. Hence, by remark \ref{rem:subsetsimulation}, $((C\setminus D') \cup D) \la (C \setminus D')$ and $(C \setminus D') \cup D \la D \la D'$. Hence, $(C \setminus D') \cup D~ \r ~((C \setminus D') \cup D')$. Hence, $((C \setminus D') \cup D) \la C$.
\end{proof}

Both the above remarks imply the following Proposition. We abuse the notation by writing $\{s\} \la \{s'\}$ as $ s \la s'$. 
\begin{proposition}
\label{prop:reduce}
    If $s, s' \in C$ and $s \la s'$ then $C \setminus \{s'\} \ca C$.
\end{proposition}
\begin{proof}
Notice that $(C \setminus \{s'\}) \cup \{s\} = C \setminus \{s'\}$.
Hence, by Remark \ref{rem:replacesimulation}, $(C \setminus \{s'\}) \la C$. By Remark \ref{rem:subsetsimulation}, $C \la C \setminus \{s'\}$. Hence proved.
\end{proof}
We use the above Proposition \ref{prop:reduce} and  Lemma \ref{lem:main} (which holds for $\vwatau$ and  1-$\atau$) to bound the cardinality of the configuration preserving simulation equivalence. This bound on the cardinality of configurations will imply that we need to remember only a bounded number of clock values to simulate these configurations. Hence, we use this bound on the cardinality of the configurations to bound the number of clock copies required while constructing the required timed automaton.

\subsection{Bounding Cardinality of Configurations}
\subsubsection{Intuition}
We now discuss the intuition for the decidability of $\vwatau$. The main reason for the undecidability of $\ata$ or $\vwata$ is due to the unboundedness of the configuration size. That is, the cardinality of the configurations could depend on the length of the timed word prefix read so far. Hence, we need to keep track of an unbounded number of clocks. This happens, because we can reset a clock $x$ in one branch and not reset $x$ in another branch while taking transitions. This is a result of transitions containing clauses of the form $(X_i.q_i \wedge X_j.q_j)$ where $X_i \ne X_j$ and $X_i, X_j \subseteq X$. That is, we get two states in the successive configuration each resetting a different set of clocks. Hence, we need to remember multiple values for clock variables that are reset in one branch and not in another. In case of $\atau$, we observe the following:
\begin{itemize}
    \item Observation 1 - Let $q \in \Qg$. Due to the nature of constraints, i.e. $x_i \in (l, \infty)$, if we have a pair of states $(q, \nu_1), (q,\nu_2)$ in a configuration $C$, such that $\nu_1 \le \nu_2$ (i.e. $\forall x \in X. \nu_1 (x) \le \nu_2(x)$), then any timing constraint that is satisfied by $\nu_1$ will also be satisfied by $\nu_2$. Hence, any transition that can be taken by $(q, \nu_1)$ can also be taken by $(q, \nu_2)$. Moreover, after taking the same transition (time delay followed by event-based transition) both $(q, \nu_1)$ and $(q, \nu_2)$ get states of the form $(q', \nu_1')$ and $(q',\nu_2')$, respectively, in their successor configurations, such that $\nu_1' \le \nu_2'$ if $q' \in \Qg$ and $\nu_1' = \nu_2' = \mathbf{0}$ if $q' \in \Ql$. Thus, $(q, \nu_1) \la (q,\nu_2)$. Hence, by Proposition \ref{prop:reduce}, we can delete $(q,\nu_2)$ from $C$ preserving simulation equivalence (and hence the language). A similar argument applies for $q \in \Ql$.  
    \item Observation 2 - In 1-ATA, for any pair of valuations $\nu_1, \nu_2$, either $\nu_1 \le \nu_2$ or $\nu_2 \le \nu_1$. Hence, on applying the reduction using Proposition \ref{prop:reduce} (and discussed in the previous bullet, i.e., Observation 1), we will always get a configuration, where each location appears at most once. Hence, the cardinality of configurations is bounded by the number of locations. 
    \item Observation 3 - But this is not necessarily the case for multiple clocks. This is because there could be unboundedly many incomparable valuations. For example, for 2-clocks $X = \{x, y\}$, consider the following family of configurations parameterized by $m$, $C_m = \{ (q, x=0.1+n k, y=0.9 -n k)|n \in \{0, \ldots,m-1\}\}$ and $k = 0.8/m$. $|C_m| = m$ and all the clock valuations are incomparable. Notice \\$C_8 = \{{(q, x=0.1, y=0.9), (q, x=0.2, y=0.8) \ldots (q, x=0.9, y=0.1)}\}$. \\Hence, as the second main step we show that, if $\Aa$ is a $\vwatau$, and if we conservatively keep on compressing the configurations as discussed in Observation 1 (using Proposition \ref{prop:reduce}), we will have boundedly many incomparable clock valuations. To be precise, we will have at most one copy of each location in the configuration. This is shown in Lemma \ref{lem:main} (the main technical  Lemma).
\end{itemize}

\subsubsection{Bounding Lemma}
\label{sec:bound}
In this section, we will use the intuition in Observation 1 for constructing a simulation equivalent transition system for a given 1-$\atau$ and $\vwatau$ whose states are configurations of given ATA $\Aa$ with bounded cardinality. For the 1-$\atau$, the intuition in Observation 2 guarantees the case. 
For the multi-clock $\vwatau$, the issues discussed in Observation 3 must be resolved. This is resolved in Lemma \ref{lem:main}, the main contribution of this section.

In what follows, assume $\Aa$ to be an $\atau$. We define relation $\lasim$ amongst states of $\Aa$.
For ${\sim} \in \{\le, \ge\}$, let $\lasim$ be defined between states such that $s \lasim s'$ iff $s=(q, \nu)$, $s' = (q, \nu')$, and if $q \in \Qsim$ then $\nu' \sim \nu$. By Observation 1 we have Proposition \ref{prop:red}. The formal proof is in Appendix \ref{app:small}. 
\begin{proposition}
\label{prop:red}
$s \lasim s'$ implies $s \la s'$. 
\end{proposition} 
Given any configuration $C$, we define $\reduce(C)$ as a configuration $C'$ obtained from $C$, by deleting all states $s' \in C'$ if there exists a state $s\in C'$, such that $s\ne s'$, and $s\lasim s'$. Intuitively, we delete some information from a configuration that is redundant in deciding whether a timed behaviour from that state is accepted or not.

Let $C_0$ be the initial configuration of $\Aa$. Let $TS(\Aa) = (\mathcal{C}, C_0, (\mathbb{R}_{\ge 0} \times \Sigma), \rightarrow_{\Aa})$ be the transition system corresponding to $\Aa$. We define $\tred(\Aa)$ as a transition system $\tred(\Aa) = (\mathcal{C}, C'_0, (\mathbb{R}_{\ge 0} \times \Sigma), \rightarrow_{\Aa, \mathsf{red}})$ such that $C'_0 = \reduce (C_0)$ and for any $C, C', D, D' \in \mathcal{C}$, $a \in \Sigma$, and $t \in \mathbb{R}_{\ge 0}$, $C \xrightarrow{(t,a)}_{\Aa} C'$ iff $D \xrightarrow{(t,a)}_{\Aa, \mathsf{red}} D'$, $D = \reduce (C)$, and $D' = \reduce(C')$. By Proposition \ref{prop:reduce2}, $TS(\Aa)$ is simulation equivalent to $\tred(\Aa)$. 
 The following Proposition is implied by Proposition \ref{prop:reduce} and \ref{prop:red}.
 
\begin{proposition}
     \label{prop:reduce2}
     $C \cong_\Aa \reduce(C)$. Hence, $TS(\Aa)$ and $\tred(\Aa)$ are simulation equivalent. 
 \end{proposition}




\begin{remark}
    \label{rem:redC}
Any run $R'$ is a run of $\tred(\Aa)$ iff $R' = \Img(R)$ for some run $R$ of $\Aa$, where $\Img(R)$ is defined as follows.
$R = {C_0 \xrightarrow{(t_0, a_0)}_\Aa C_1 \xrightarrow{(t_1, a_1)}_\Aa C_2 \ldots}$, we define $\Img(R)$ as run 
\\$R' = C_0''\xrightarrow{(t_0, a_0)}~~C_1'' \xrightarrow{(t_1, a_1)}~~C_2'' \ldots$ where $C_0' = C_0'' = \reduce(C_0)$ and $\forall i \ge 0. C''_{i} \xrightarrow{(t_{i}, a_{i})}_\Aa C_i'$ and $C_i'' = \reduce C'_i$.
\end{remark}

\begin{lemma}
\label{lem:main}
Let $\Aa = (Q,\Sigma, \delta, q_0, \Qacc, X, \G)$ be either an 1-$\atau$ or $\vwatau$ . Let $R$ be a run of $\Aa$, and $R' = \Img(R) =C_0'' (t_0, a_0)C_1'' (t_1, a_1)\ldots $, then for all $i \ge 1$, $C_i''$ does not contain states $(q,\nu)$ and $(q, \nu')$ where $\nu \ne \nu'$ for any $q \in Q$. In other words, every location $q \in Q$ appears at most once in any configuration $C_i''$ for any $i \ge 1$. Hence, $|C_i''| \le |Q|$.
\end{lemma}
\begin{proof}[Proof (sketch)]
Notice that if $\Aa$ was 1-$\atau$, the above statement is straightforward as no two clock valuations are incomparable in the case of 1-clock. We now show the same for $\Aa$ being a multi-clock $\vwatau$. We just present intuition behind the proof idea. A formal proof is proved using DAG semantics of ATA and can be found in Appendix \ref{app:main}.
We prove this by contradiction. Assumption 1 - Suppose $k$ is the smallest number such that $C_k''$ contains two copies of some location $q \in Q$. Hence, there exists $\nu $ and $\nu'$ such that $\nu'$ is incomparable to $\nu$ and $(q,\nu), (q,\nu')\in C_k''$. Then, the following cases are possible:
\\Case 1 - Both $(q,\nu), (q,\nu')$ appeared from the same location $p$ in $C''_{k-1}$. But, by condition (3) of $\vwata$,
all the transitions from location $p$ to location $q$ reset the same set of clocks. Moreover, by assumption 1, location $p$ appears at most once in $C''_{k-1}$. Let $(p,\nu_p) \in C''_{k-1}$. Then both the clock valuations $\nu$ and $\nu'$ should be identical as they result from the same state $(p, \nu_p)$ resetting the same set of clocks.
\\Case 2 - $(q,\nu), (q,\nu')$ appeared from distinct location $(p, \nu_{k-1})$ and $(p', \nu'_{k-1})$ in $C''_{k-1}$. By condition (3) of $\vwata$ there is at most one location $q' \ne q$ from which there are transitions entering location $q$. Moreover, all these transitions reset the same set of clocks.  Hence, one of $p$ and $p'$ has to be $q$. Wlog $p = q$. It suffices to show that 
whenever such a case occurs, the clock valuation of the state that results from the self-loop (in this case $\nu$) is always greater than or equal to the valuation from the other (in this case $\nu'$) (Statement 1). Hence, $\nu' \le \nu$ which leads to a contradiction. We just present the intuition with an example. Let $\rho = (a_1,  \tau_1), (a_2, \tau_2)$.
Suppose, $(q_0,\mathbf{0})$ is the initial location of the automaton as drawn in Figure \ref{fig:ta}. Let $k=2$. Notice the run in the Figure, $C_1 = \{(q_0, \nu_1), (q_1, \nu_1')\}$ where if $x \in X'$, $\nu_1' (x) = 0 \le \nu_1(x) = \tau_1$. Else, $\nu_1(x) = \nu_1'(x) = \tau_1$. Similarly, $C_2 = \{(q_0, \nu_2), (q_1, \nu'), (q_1, \nu)\}$, where $(q_1, \nu)$ results from the self loop and $(q_1, \nu')$ results from the transition from $q_0$. Hence, if $x = X'$, $\nu (x) = 0 \le \nu'(x) = \tau_2 - \tau_1$. Else, $\nu_1(x) = \nu_1'(x) = \tau_2$. In other words, while reaching both $(q_1, \nu)$ and $(q_1, \nu')$ from the initial configuration, the same set of clock $X'$ was reset. But, in the case of the former, they were reset before the latter. Hence, $\nu$ and $\nu'$ agree on all the clock values not in $X'$ and $\nu \ge \nu'$ for all the clocks in $X$. Applying this argument inductively we can prove Statement 1. We believe it is more intuitive to prove the result using the DAG semantics of ATA. Hence, the full proof can be found in the Appendix \ref{app:main}, where we introduce the semantics too.
\end{proof}

\begin{figure}[h]

\includegraphics[scale=0.25]{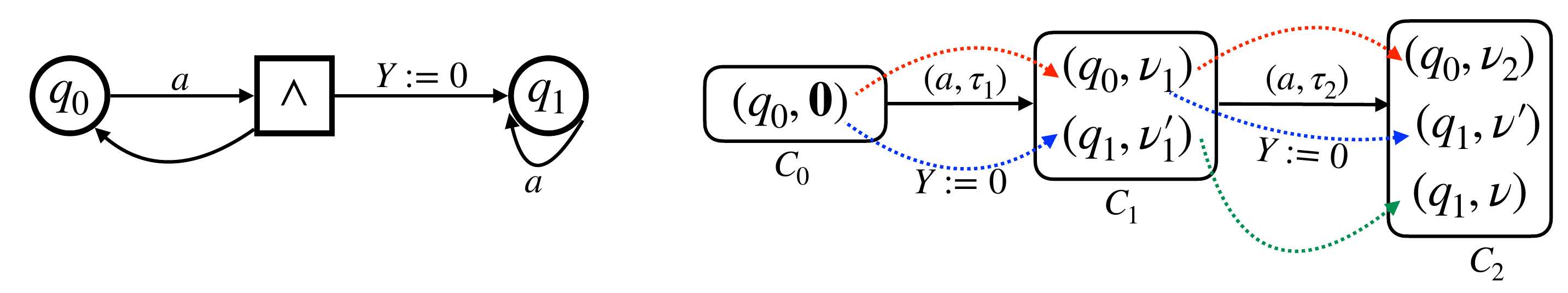}

\caption{The red and green transitions denote those without resets, and the blue ones with resets. Notice the paths from $C_0$ to $C_2$. The Blue-Green and Red-Blue path reset the same set of clocks $Y$. But the former resets the clocks earlier (in the first step) as compared to the latter (in the second step). Hence in the former, clocks in $Y$ get a chance to progress between $C_1$ and $C_2$. Moreover, both the paths should agree on the value of clocks not in $Y$ as they are not reset in both these paths. Hence, $\nu' \leq \nu$.}
        \label{fig:ta}
\end{figure}
\subsection{From $\vwatau$ to Timed Automata}
\label{sec:ata-ta}
In this section, we propose an on-the-fly construction from $\vwatau$ to Timed Automata.  The termination relies on  Lemma \ref{lem:main}. The main idea is to bind the number of active clocks using  Lemma \ref{lem:main}. Given a $\vwatau$ or $1-\atau$, $\Aa = (Q, \Sigma, \delta, q_0, \Qacc, X, \G, \Qg, \Ql)$ we get a timed automaton $\Aa = (\Q, \Sigma, \Delta, q'_0, \Qc, X \times \{0,\ldots, |Q|-1\}, \G)$ and at every step we reduce the size of the location $q \in \Q$ preserving simulation equivalence. Let $V$ be set of all the functions of the form $v: X \mapsto \{0,\ldots, |Q|-1\}$. Let $\lo$ be a set of all the functions from $Q$ to $V\cup \{0\}$. Let $\act$ be a set of all the functions from  $X$ to a sequence (without duplicate) over $\{0,1\ldots |Q-1|\}$. Then $\Q = \lo\times \act$.
Intuitively, we replace the bunch of conjunctive transitions $C$ into a single transition, similar to the subset construction for converting Alternating Finite Automata (AFA) to Non-Deterministic Finite Automata (NFA). But notice that we can have clauses (or conjunctions) of the form $q \wedge X'.q'$. Hence, simple subset construction won't work as we need to spawn multiple copies of clocks in $X'$, wherein one of the elements of the new location $\{q, q'\}$ they are reset while in another they are not. In general, there could be an unbounded number of such clock copies required for a single clock, $x \in X$. But due to  Lemma \ref{lem:main}, if we make sure to compress the states (and hence remove redundant clocks), we need to keep at most $|Q|$ copies for each clock in $X$. In principle, we are constructing an NTA $\A$ whose transition system $TS(\A)$ is simulation equivalent to the LTS $\tred (\Aa)$ (see Appendix \ref{app:atatonta} Proposition \ref{prop:tasim}) and hence to input $\vwatau$ $TS(\Aa)$. Thus, by Proposition \ref{prop: langeq}, $L(\Aa) = L(\A)$.  We present the idea via our running example. 
\begin{figure}[t]
    \includegraphics[scale=0.2]{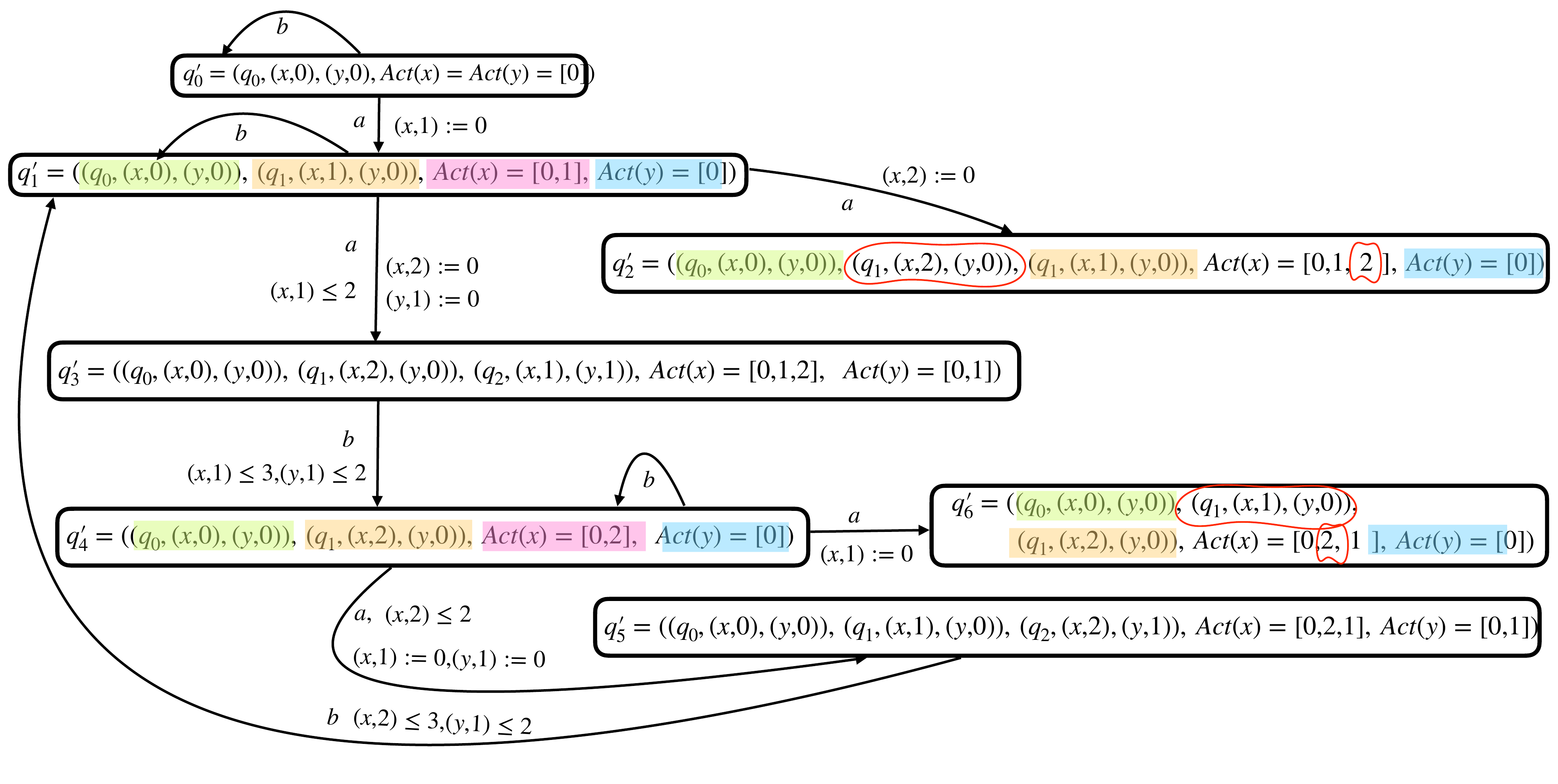}    
    \caption{Steps in the construction of $\A$ corresponding to our running example. 
     With the  color coding in $q'_1,q'_2$, it is easy to see that  $q'_2$ is same as $q'_1$ on removing the circled entities in $q'_2$. Same with $q'_4$ and $q'_6$.}
            \label{fig:my_label}
\end{figure}
\subsubsection{Construction on Running Example}
\label{sec:exc}
Please refer to the $\vwatau$ of our running example Figure \ref{fig:running-example}. We now illustrate the construction on our running example.
We start with location $q_0$, with the $0^{th}$ copies of clock $x$ and $y$. Hence \\$q'_0 = \{(q_0, (x,0) (y,0)), \act(x) = [0], \act(y) = [0]\}$. This corresponds to the configuration $C_0 = \{(q_0, \mathbf{0}_X)\}$ of $\Aa$. In the input automaton, the transitions from $q_0$ on $a$ is defined by $\delta(q_0, a) = q_0 \wedge x.q_1$. 
Hence, we need to spawn a new copy of clock $x$ as it is reset in one transition and not in another. We associate this new copy of clock $x$ with the branch that resets $x$, i.e., this new clock $x$ is associated with location $q_1$. Hence, we  have $\Delta (q'_0, a) = (x,1).q'_1$ where $q'_1 = \{(q_0, (x,0), (y, 0)), (q_1, (x,1) (y,0)), \act(x) = 0 \ge 1, \act (y) = 0\}$. Intuitively, $q'_1$ corresponds to the configurations of the form $C_1 = \{(q_0, \nu^{1}_{0}), (q_1, \nu^{1}_{1})\}$ of $\Aa$, where $\nu^{1}_{0}(x) = \text{value of }(x,0)$, $\nu^{1}_{1}(x) = \text{value of }(x,1)$, and $\nu^{1}_{0}(y) = \nu^{1}_1(y) = \text{value of }(y,0)$. 
We continue with this new location. Hence, we will consider the transitions from both $q_0$ and $q_1$ on $a$. The component $((q_0,(x,0),(y,0)))$ on $a$  again spawns a new
copy of clock $x$ as it resets the clock in one while not resetting on self-loop, hence, getting $\{(q_0,(x,0),(y,0)), (q_1, (x,2) (y,0))\}$ (possibility 1 from $q_0$, the only possibility). Notice that we spawned $(x,2)$ as $(x,1)$ is in use
by $q_1$ already. The component $(q_1, (x,1) (y,0))$ will be computed using the transition function of input automaton, i.e. $\delta (q_1, a) = (y.q_2 \wedge x \le 2) \vee q_1$. Here,  we either stay at $q_1$ with the same set of clock copies as before (possibility 1 
from $q_1$), or  we need a new copy of $y$ while simultaneously checking for the clock copy of $x$ 
corresponding to location $q_1$ (i.e. $(x,1)$) is $\le 2$ (possibility 2 from $q_1$). Combining the possibilities 1 from $q_0$ and $q_1$ we get, $\scriptstyle{{\{(q_0, (x,0), (y,0), (q_1, (x,2) , (y,0))(q_1, (x,1), y,0), 
\act(x) = [0 \ge 1 \ge 2], \act(y) = [0] \}}}$. But $q \in \Ql$. Hence, if we can reach the accepting state from $(q_1, (x,1) , (y,0))$ then we can reach the accepting state from $(q_1, (x,2), 
(y,0)$ too, as $\text{value of }(x,1) \ge \text{value of }(x,2)$ (this fact is also encoded in the $\act(x)$ sequence). Thus, $(q_1, (x,2), (y,0))$ can be removed from the new location without 
affecting simulation equivalence (and hence language equivalence). This corresponds to the removal of redundant states in the construction of the runs of $\tred(\Aa)$ from $T(\Aa)$. Hence, after deletion we get $\scriptstyle{{q'_2 = \{(q_0, (x,0), (y,0), \cancel{(q_1, (x,2) , (y,0))},(q_1, (x,1), y,0), \act(x) = [0 \ge 1 \cancel{\ge 2} ], \act(y) = [0] \} = q'_1}}$.

 \begin{figure}
 \begin{center}

    \includegraphics[scale=0.25]{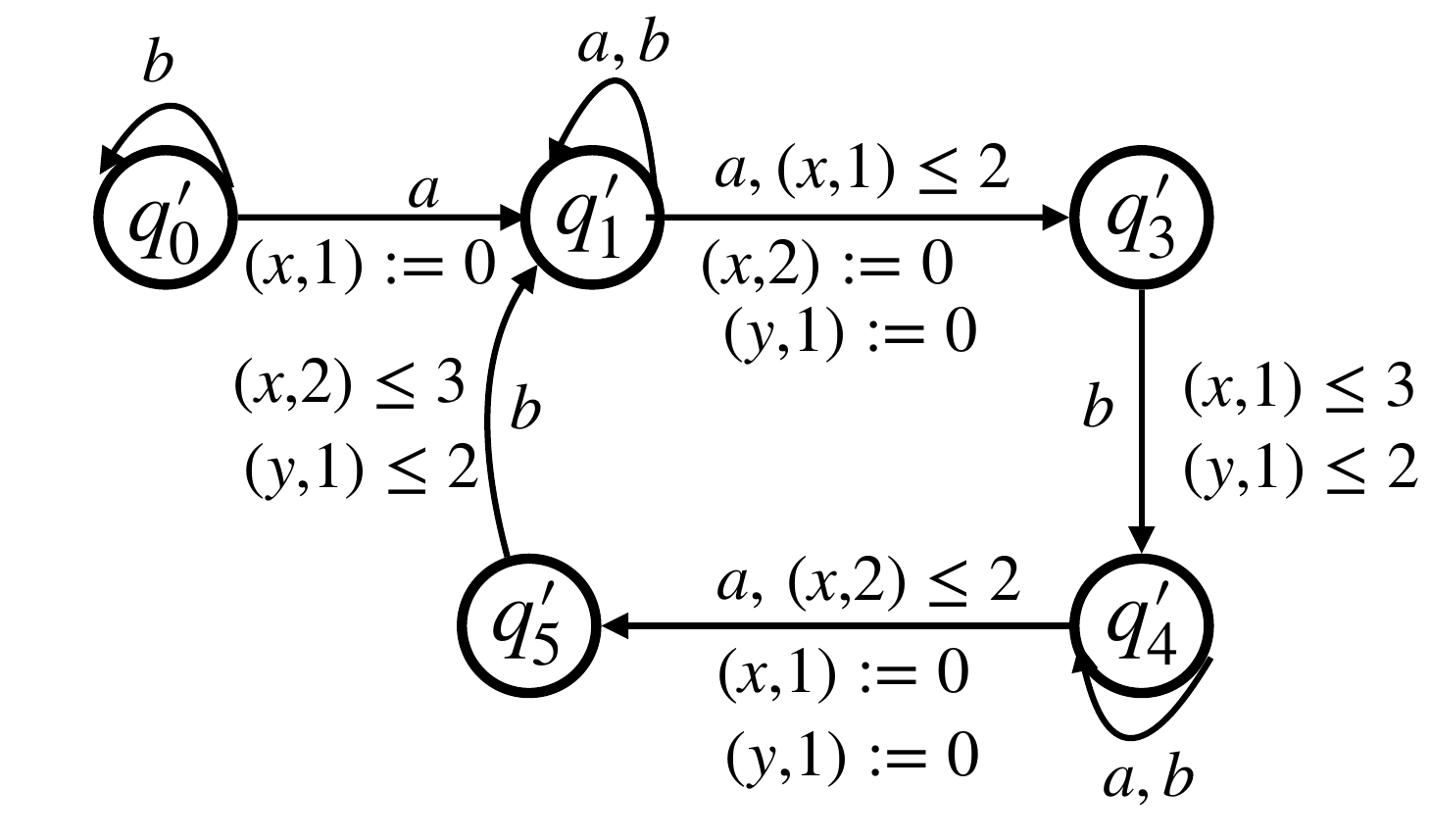}
    \caption{  Final Automata after applying the reductions}
    \label{fig:c}
         
 \end{center}
\end{figure}

Thus, combining result of the transition of $q_0$ on $a$ and possibility 1 from $q_1$ we get $\scriptstyle{{q'_1 =\{q_0, (x,0), (y,0), (q_1, (x,1), y,0),  \act(x) = 0 \ge 1, \act (y) = 0\}}}$.\\ 
Combining results possibility 1 from $q_0$ and possibility 2 from $q_1$, we get \\$\scriptstyle{\{q_0, (x,0), (y,0), (q_1, (x,2), (y,0)),}$ $\scriptstyle{(q_2, (x,1), (y,1))\act(x) = 0 \ge 2 \ge 1, \act (y) = 0 \ge 1.\} = q'_3}$ if $(x,1) \le 2$. Note that each location from $Q$ appears at most once in $q'_3$. Hence, there is no scope of reduction. Combining the above two combination of possibilities, $\Delta(q'_1, a) = q'_1 \vee (y,1).(q'_3 \wedge x \le 2)$. Continuing this we get the resulting NTA $\A$ equivalent to the input formula $\phi$. Notice that we are eliminating the conjunctive transitions using subset like construction and keeping the disjunctions as it is. Hence, after eliminating all the conjunctive transitions the reduced automata contains only disjunctions amongst different locations in the output formulae of the transitions giving an NTA.
Refer to Figures \ref{fig:my_label}, \ref{fig:c}.

%
\subsubsection{Worst Case Complexity}
By construction in \cite{AD94}, the number of states in the region automata of $\A = W \le |\Q| \times (|X| \times |Q|)! \times 2\times (c_{max} + 1)$ where $c_{max}$ is the max constant used in the guards in $\G$ and $|\Q| = |Q| \times (|X|^{|Q|}+1) \times (|Q|!)^{|X|}$. Hence, $W = O(2^{Poly(|A|)})$ implying that the emptiness could be checked in NPSPACE = \pspace. Notice that the state containing the location $(L,\ac)$ will only have to store the region information of active clocks, which, in practice, could be much less than the worst case. Hence, lazily spawning clock copies may result in NTA with much less number of clocks than the worst case (i.e. $|X| \times |Q|$).


\section{Expressiveness of $\tptlu$}
\label{sec:exp}
We now compare the expressive power of 1-$\tptlu$ with respect to  that of $\mitl$. 

\subsection{Metric Temporal Logic(MTL)}
\label{sec:mE}
$\mtl$ is a real-time extension of $\ltl$ where the $\until$ modality is guarded with an interval. 
Syntax of $\mtl$ is defined as follows. $\varphi::=a ~|\top~|\varphi \wedge \varphi~|~\neg \varphi~|
~\varphi \until_I \varphi,$
\\where $a \in \Sigma$ and  $I \in \intinterval$.    
For a timed word $\rho = (\sigma_1, \tau_1 ) (\sigma_2, \tau_2) \ldots (\sigma_n, \tau_n) \in T\Sigma^*$, a position 
$i \in dom(\rho)$, an $\mtl$ formula $\varphi$, the satisfaction of $\varphi$ at a position $i$ 
of $\rho$, denoted $\rho, i \models \varphi$, is defined as follows. We discuss only the semantics of temporal modalities. Boolean operators mean as usual.
$\rho,i\ \models\ \varphi_{1} \until_{I} \varphi_{2} \ \text{iff} \  \exists j > i. 
\rho,j\ \models\ \varphi_{2}, \tau_{j} - \tau_{i} \in I,\ \text{and} \  \forall  i< k <j. \rho,k\ \models\ \varphi_{1}$.
As usual, $\fut_I(\phi) = \top \until_I \phi$, $\sbf(\phi) = \neg\fut_I \neg \phi$ , $\nx_I \phi = \bot \until_I \phi$.
The language of an $\mtl$ formula $\varphi$ is defined as $L(\varphi) = \{\rho | \rho, 1 \models \varphi\}$.
The subclass of $\mtl$ where the intervals $I$ in the until
modalities are restricted to be {\bf non-punctual}  is known as Metric Interval Temporal Logic ($\mitl$) . $\mitlu$ \cite{AFH96}\cite{AlurFH91}\cite{icalp-raskin} is the subclass of $\mtl$ where intervals are restricted in $\intintervalu$. Satisfiability Checking for $\mitl$ ($\mitlu$) is \expspace-complete (\pspace-complete) \cite{AlurH92} \cite{AFH96} \cite{AlurFH91}. $\mitl$ is strictly more expressive than $\mitlu$ in pointwise semantics \cite{H98}. 
\begin{remark}
\label{rem:mitl-tptl}
Any $\mtl$ formula can be translated to an equivalent 1-$\tptl$ (closed) formula using the following equivalence recursively. 
$\varphi_1 \until_I \varphi_2 \equiv x.(\varphi_1 \until \varphi_2 \wedge x \in I)$. 
\end{remark}
\subsection{Expressiveness of $\tptlu$}
\begin{theorem} 1-$\tptlu$ is strictly more expressive than $\mitl$. 
\end{theorem}
\begin{proof}
 Both $\mitl$ and 1-$\tptlu$ are closed under all boolean operations. Hence, we just need to show that any formula of the form  $\varphi' \until_{I} \varphi$ 
 is expressible in 1-$\tptlu$. 
Notice that, any $\mitl$ formula  $\varphi' \until_{[l,u)} \varphi \equiv [\sbf_{[0,l)} \{\varphi' \wedge (\varphi' \until \varphi)\}] \wedge [ \fut_{[l,l+1)} \varphi \vee \fut_{[l+1,l+2)} \varphi \ldots \fut_{[u-1,u)} \varphi]$. (similar reduction applies for other kinds of intervals).  $\sbf_{[0,l)} (\varphi' \wedge (\varphi' \until \varphi))$ is already in $\mitlu$ (and hence in 1-$\tptlu$ by remark \ref{rem:mitl-tptl}). 
Hence, it suffices to encode modalities of the form $\fut_{[l,l+1)}$ using 1-$\tptlu$ formula.    
Let $\rhos = (a_1, \tau_1), (a_2, \tau_2) \ldots$ be any timed word. Let $i \in dom(\rhos)$ be any point. $\rhos, i \models \fut_{[l,l+1)}(\varphi)$ iff there exists a point $i' > i$ such that $\tau_{i'} - \tau_i \in  [l, l+1)$ and $\rhos, i' \models \varphi$. $\rhos$ has a point $i'$ within $[l, l+1)$ interval from $i$ where $\varphi$ holds iff there exist earliest such point $j$ ($j \le i'$) within $[l, l+1)$ from $i$ where $\varphi$ holds iff there is a point $j' > i$ such that $\tau_{j'} - \tau_i \ge l$ (i.e. $\rho, i \models \phi_0  = \fut_{[l, \infty)} \varphi$), and let $j$ be the first point such that $\tau_{j} - \tau_{i} \ge l$, and $\rhos, j \models \varphi$. Such a point exists due to $\phi_0$. Then:
\begin{itemize}
    \item Case 1: Either there is no point strictly between $i$ and $j$ where $\varphi$ holds. Then occurrence of $j$ within $l+1$ can be expressed using formula, $\phi_1 = \neg \fut_{[0, l)} \varphi \wedge \fut_{[0,l+1)} \varphi$.
    \item Case 2: Or there exists a point $k$ such that $\tau_j - \tau_k < 1$, $\tau_k - \tau_i \in [l-1, l)$, and $\rhos, k \models \varphi$. Equivalently, $i$ satisfies $\phi_2 = \sbf_{[l-1, l)}(\fut_{[0,1]} (\varphi))$,
\item Case 3: Or there exists a point $k$ with $i < k <j$ such that $\tau_j - \tau_k \ge 1$, $\rhos, k \models \varphi$, and $\forall k < k' < j. \rho, k' \not\models \varphi$. \\
(1) Such a point $k$ satisfies 
    $\phiapproach = \varphi \wedge \sbf_{[0,1)}(\neg \varphi)$.
    Indeed a key property is that $k$, the last point in $[0,l)$ satisfying $\phi$, satisfies $\phiapproach$. By the definition of $k$, i.e., there are no occurrences of $\varphi$ after $k$ in $[0,l)$. \\ 
 (2) Notice that any two point $k_1$ and $k_2$ satisfying $\phiapproach$ are at least a unit time apart. Hence, there could be at most $l$ points satisfying $\phiapproach$ within  $[0, l)$. Then, the following 1-$\tptlu$ formula
    $\cntphi(n)$ with parameter $n$ states that there are exactly $n$ points, $1 \leq n \leq l$ within $[0, l)$ of point $i$ where $\phiapproach$ holds. Here, $\cntphi(n)~ = ~\phi_{\ge n} \wedge \neg \phi_{\ge n+1}$, where $\phi_{\ge n } =  x.((\neg \phiapproach) \until (\phiapproach \wedge $ \\
    \quad $ ((\neg \phiapproach )\until (\phiapproach \wedge \underbrace{\ldots}_{n-3} \wedge ((\neg \phiapproach) \until (\phiapproach \wedge x < l)\ldots)))$. \\
    Observe  that for a given timed word and interval $[0,l]$ from $i$, there is a unique $n$ satisfying this formula $\cntphi(n)$. \\
   (3)  Using this $n$, the formula 
       $\gamma(n,\varphi)= x.(\neg \phiapproach \until (\phiapproach \wedge \neg \phiapproach \until (\phiapproach \wedge \underbrace{\ldots}_{n-3} \wedge ((\neg \phiapproach) \until (\phiapproach \wedge \sbf(x \leq l \lor \neg \varphi) \wedge \fut(\varphi \wedge x < l+1))\ldots))$
          holds if after $n$ occurrences of $\phiapproach$ (which gives point $k$), the next occurrence of $\varphi$  occurs before time $l+1$.  \\
Hence, case 3 is characterized by the formula $\phi_3 =  \bigvee \limits_{n=1}^{l}  \cntphi(n) \wedge \gamma(n,\varphi)$.
    \end{itemize}
     Hence, the required formula $\psi = \phi_0 \wedge (\phi_1 \vee \phi_2  \vee \phi_3)$.
 For strict containment of MITL, consider the formula $\beta = x.\fut(b \wedge \fut(b \wedge x \le 1))$. This specifies, there exist at least two points within the next unit interval where $b$ holds. \cite{rabinovichY} \cite{count} \cite{khushraj-thesis} show that this formula is not expressible in $\mtl$ (and thus also in the subclass $\mitl$). 
\end{proof}



\section{Discussion and Conclusion}
\label{sec:discuss}
Ferr{\`{e}}re \cite{F18} proposed an extension of LTL with Metric Interval Regular Expressions called Metric Interval Dynamic Logic (MIDL) and showed it to be more expressive than $\emitl$ of \cite{Wilke}. We claim that our proof of \pspace completeness for 1-$\atau$ emptiness implies the same for MIDL$_{0, \infty}$ satisfiability strictly generalizing the results and techniques of \cite{H19} which proved the same for $\emitl_{0,\infty}$. This resolves one of the `future directions' of \cite{H19}. Authors in \cite{KKMP21} generalized the notion of non-punctuality to non-adjacency for 1-$\tptl$. We remark that unfortunately, this notion doesn't help in making 2-TPTL decidable. 
Notice that $\varphi =  \sbf x.\{\neg \phi \vee \fut y.(\top \wedge x \in [1,2] \wedge \fut(\phi_1 \wedge x \in [1,2] \wedge y \in [1,2]))\} \equiv \sbf[\phi \rightarrow \fut_{[1,1]}(\fut_{[1,1]}\phi_1)]$. Because, for any point $i$ where $\varphi$ holds there is a point $j$ in the future such that $\tau_{j} - \tau_i \in [1,2]$, and from that point $j$ there is a point $k$ in the future where $\phi_2$ holds such that $\tau_{k} - \tau_{j} \in [1,2]$ and  $\tau_{k} - \tau_i \in [1,2]$. Solving the inequalities we get, $\tau_{j} - \tau_i  =  1$ and $\tau_{k} - \tau_i = 2$. Hence, $\varphi$ can express some restricted form of punctual timing properties which leads to the undecidability of satisfiability using encoding similar to \cite{OuaknineW06}. 
$\mitlu$ was extended with Counting (TLC) and Pnueli (TLP) modalities by \cite{rabinovichY} to increase the expressiveness, meanwhile maintaining the decidability in EXPSPACE and PSPACE, respectively. TLP and TLC have the same expressive power. This is because expressing arbitrary non-punctual metric interval constraints using unilateral intervals is a non-trivial phenomenon in pointwise semantics (see \cite{H19}).
While these logics were strictly more expressive than $\mitl$ in continuous semantics in pointwise they are incomparable.  Moreover, TLP and TLC properties are trivially expressible in $\tptlu$ (see construction in Appendix K.1 in \cite{count-full}), making our logic strictly more expressive than these. 
As one of our future work, we would like to show that TLCI and TLPI (extensions of TLP and TLC using arbitrary non-punctual intervals) which are decidable in EXPSACE are expressible in $\tptlu$.
Finally, we leave open (i)the extension of this work with Past modalities, (ii)FOL-like characterizations of $\tptlu$, and (iii) whether adding multiple clocks in $\tptlu$ improves expressiveness.

\bibliographystyle{plain}
\bibliography{papers}

\newpage
\appendix
\section{Reduction from TPTL to $\vwata$}
\label{app:tptltoata}
\subsection{Pre-processing}
We call a formula $\phi$ as temporal formula iff the topmost operator is a temporal modality. Similarly, we call a subformula $\phi'$ of $\phi$ as a temporal subformula. 

\noindent{\bf{Pushed Formulae}}. Given any TPTL formula $\phi$, we first distribute all the freeze quantifiers, such that, only a single temporal formula appears within the scope of a freeze quantifier. This can be done using the identity recursively: $x.(\phi \wedge \phi') = x .\phi \wedge x.\phi'$, $y.(\phi \vee \phi') = y.\phi \vee y. \phi', x.(Prp) = Prp, x.(x \in I)= \top$ if $0 \in I$, else $\bot$, where $Prp$ is a propositional formula over $\Sigma$ and clock constraints not containing the variable $x$. We call such a formula where quantifiers are pushed inside as a \emph{pushed formula}. For example, $\phi = x. (a \until y.(b \wedge y \in (2, 3) \vee x \in (1,2)) \wedge a \wedge x \in [0, 1)) = x.(a \until y.(b \wedge y \in (2, 3) \vee x \in (1,2)))  \wedge a \wedge \top =  x.(a \until(b \wedge \bot \vee x \in (1,2)))  \wedge a \wedge \top = x.(a \until(x \in (1,2)))  \wedge a \wedge \top$.

\smallskip 

\noindent{\bf{Strictly Closed Formulae}}. Similarly, given any closed formula $\phi$, $X.(\phi)$ is equivalent to $\phi$, by semantics of $\tptl$ for any $X$. Given any formula $\phi$ using clocks in $|X|$, change every closed subformula $\phi'$ of $\phi$ to $X.\phi'$. 
$X.\phi'$ is called a \emph{strictly closed} formula. For example, $\phi = x.y.(a \wedge x.\fut(a \wedge x \in (0,1)))$. Notice that subformula $x.\fut(a \wedge x \in (0,1))$ is closed. Hence, changing it to $x.y.\fut(a \wedge x \in (0,1))$ would not affect the language of the formula, $\phi$. Hence, the new formula we get is $x.y.(a \wedge x.y.\fut(a \wedge x \in (0,1)))$. We can now assume that the given formula $\varphi$ for which we want to construct $\vwata$ is a pushed, and  strictly closed formula. 

\smallskip

Let $A \in 2^\Sigma$ and $\phi$ be any pushed strictly closed formula. We first define $\formula(\phi, A)$ inductively as follows. $\formula(\phi_1 \wedge \phi_2, A ) = \formula(\phi_1) \wedge (\phi_2)$, $\formula(\phi_1 \vee \phi_2, A) = \formula(\phi_1) \vee (\phi_2)$, $\formula(X.\phi_1, A) = X. \formula(\phi_1, a)$, $\formula (a, x \in I) = x \in I$, $\formula(a, \phi)  = \phi)$ if $\phi$ is a temporal formula, 
$\formula(a, A) = \top$ if $a \in A$, else $\formula(a, A) = \bot$.   

Let $\cl(\varphi)$ contain all the temporal subformulae of $\varphi$.
Given any $\tptl$ formula we now define the ATA $\Aa_\varphi = (Q, 2^\Sigma, \delta, q_0, \Qacc, X, \G)$, where $Q = \cl(\varphi) \cup \varphi^0$ (we call $\varphi^0$ the initial copy of the given formula $\varphi$), $q_0 = \varphi^0$, $\Qacc$ all the formulae of the form $\sbf(...)$, i.e. formulae whose topmost operator is a $\sbf$ operator, $X$ is set of all the clocks appearing in $\varphi$, $\G$ is the set of guards appearing in $\varphi$. We define $\delta$ as follows: for any $A \in 2^{\Sigma}$, $\delta (\varphi^0, A)) = \varphi$, $\delta(\phi_1 \until \phi_2) = \formula(\phi_2,A) \vee (\formula(\phi_1,A) \wedge \phi_1 \until \phi_2)$, $\delta(\sbf(\phi), A) = \formula(\phi, A) \wedge \sbf(\phi)$. The equivalence is due to \cite{Ouaknine05}\cite{OWH}. The proof goes through mutatis mutandis for this case too.

We have the following points. 
\begin{itemize}
\item The constructed $\Aa_\varphi$ is indeed an ATA satisfying condition (1) and (2) of the $\vwata$. We call this $\vwata$(1-2).
\item As the given $\varphi$ is pushed and strictly closed, every closed subformula $X'.\psi$ would be such that $\psi \in \cl(\varphi)$ and $X' = X$. Hence, all the transitions entering the location corresponding to $\psi$ in $\Aa_{\varphi}$ are strong reset transitions. This condition will make sure that if $\varphi \in \tptlu$ then $\Aa_{\varphi}$ is indeed $\atau$.
\item The number of transitions entering a subformula $\psi$ of $\varphi$ is exactly the number of times (say $i$) subformula $\psi$ occurs in $\varphi$. Notice that, due to the structure of the automaton, $i$ is also the number of paths from initial state $\varphi^0$ to $\psi$ in $\Aa_\varphi$. Hence, we unfold the automaton, by making $i_q$ copies for every location $q$ where $i_q$ is the number of transitions entering $q$. This unfolding is guided by the tree representation of the given $\varphi$. Hence, the final automaton will mimic the structure of the tree representation of the given formula $\varphi$. After this transformation, we make sure that between any two states there is a unique path (modulo self-loops). This results in the equivalent automaton $\Aa$ which is a $\vwata$ (i.e. it satisfies all the conditions (1-3) of $\vwata$ mentioned in section \ref{sec:ata}). Notice that the measure of the size that we consider in this paper is proportional to the tree representation and not the DAG. This gap in succinctness is visible while unfolding the automaton to satisfy condition 3. 
\item Hence if the given formula $\varphi$ is a $\tptlu$ formula, then $\Aa$ is $\vwatau$ as per our definition.
\end{itemize}

\section{Proofs in section \ref{sec:vwatauemptiness}}
\label{app:vwatau}

\subsection{Proof of Remark \ref{rem:subsetsimulation}}

\label{app:subsetsimulation}

\textbf{Statement} $C \supseteq C'$ implies $C \la C'$. Hence, for any timed word $\rho$, if $\rho, i \models \Aa, C$ then $\rho, i \models \Aa, C'$. 

\begin{proof}
    We show that $\supseteq$ is a simulation relation. That is, $C’ \subseteq C$ (or $C \supseteq C’$) implies $C \la C’$. We just need to prove that $\supseteq$ indeed satisfies condition 2 (condition 1 is trivially satisfied) of the simulation relation. Notice that the set of accepting configurations is downward closed.
That is, for any accepting configuration $D$, all of its subset $D’$ is accepting (by definition of accepting configuration). Hence,  $\supseteq$ satisfies condition 2.1.

Wlog, let $C = \{s_1, \ldots s_n\}$ and $C’  = \{s_1, \ldots s_m\}$ for some $m \le n$.
For any $(t,a) \in \R \times \Sigma$, by definition, any $C \xrightarrow{(t,a)} D$ iff $D = D_1 \cup D_2 \ldots D_n$ where $D_i$ is some successor of $s_i$ on $(t,a)$. Similarly, $C' \xrightarrow{(t,a)} D' iff D’ = D_1 \cup D_2 \ldots D_m$ where $D_i$ is some successor of $s_i$ on $(t,a)$. Hence, For any $(t,a) \in \R \times \Sigma$, for any $C \xrightarrow{(t,a)} D$ there exists a $C' \xrightarrow{(t,a)} D'$ such that $D \supseteq D’$ (for any state $s_i \in C \cap C’$ , just choose the same successor of $s_i$ which was used in $D$ to construct $D’$). Hence, for any successor $C$, we have a successor of $C’$ which is a subset of that of $C$. Hence, $\supseteq$ satisfies condition 2.2.

Intuitively, these extra states $s_{m+1}\ldots s_n$ will generate extra states in the successor configurations, which will make reaching an accepting configuration harder from $C$ (and its successors) as compared to $C’$ (and its successors),
\end{proof}

\subsection{Proof of Remark \ref{rem:replacesimulation}}
\label{app:replacesimulation}

\textbf{Statement} If $D' \subseteq C$ and $D \la D'$,  then $(C\setminus D') \cup D \la C$. In other words, we can replace the states in $D'$ with that in $D$ in any configuration $C$, and get a configuration that is simulated by $C$. Hence, $L(\Aa, (C\setminus D') \cup D) \subseteq L(\Aa, C)$. 

\begin{proof}
Consider a relation $\r$ amongst configurations of $A$ such that $E~{\r}~E'$ iff $E' = E_1 \cup E_2$ such that $E \la E_1$ and $E \la E_2$. We first show that $\r$ is a simulation relation. We need to show that $\r$ satisfies condition (2.2) as it trivially satisfies condition (2.1). Let $E~\r~E'$ and $E' = E_1 \cup E_2$ such that $E \la E_1$ and $E \la E_2$.  Let $E_2' = E_2 \setminus E_1$. Then $E \la E_2 \la E_2'$ (by remark \ref{rem:subsetsimulation}). By definition of simulation preorder, for any $E \xrightarrow{(t,a)} F$ there exists (1)$ E_1  \xrightarrow{(t,a)} F_1$ and (2)$E_2' \xrightarrow{(t,a)} F_2$ such that $F \la F_1$ and $F \la F_2$. Let $F' = F_1 \cup F_2$. By semantics ATA, (1) and (2) imply $E' \xrightarrow{(t,a)} F'$.   
Moreover, by definition of $\r$, $F~\r~F'$ if $E~\r~E'$. Hence, for any $E~\r~E'$, and $E \xrightarrow{(t,a)} F$ there exists $E' \xrightarrow{((t,a)} F'$ such that $F~\r~F'$. Hence, $\r$ satisfies condition 2.2.

Given $D' \la D$. Hence, by remark \ref{rem:subsetsimulation}, $((C\setminus D') \cup D) \la (C \setminus D')$ and $(C \setminus D') \cup D \la D \la D'$. Hence, $(C \setminus D') \cup D ~\r ~((C \setminus D') \cup D')$. Hence, $((C \setminus D') \cup D) \la C$.
\end{proof}

\subsection{Proof of Proposition \ref{prop:red}}
\label{app:small}
\textbf{Statement} - $s \lasim s'$ implies $s \la s'$.

\begin{proof}
\todo{Check the proof}
We extend the relation $\preceq$ between configurations of $\Aa$ as follows. We say that $C \lasim C'$ iff $\forall (q,\nu') \in C' \exists (q,\nu) \in C$ such that if $q \in \Qsim$ then $\nu' \sim \nu$.
Let $C \preceq C'$. 
Let $f: C' \mapsto 2^C\setminus \emptyset$, such that for any $(q,\nu') \in C'$ and $q \in \Qsim$, $f((q, \nu')) = \{(q,\nu) | (q, \nu) \in C \wedge \nu' \sim \nu\}$.
Let $s' = (q,\nu')$ be any state in $C'$ and $s = (q, \nu)$ be any state in $f((q,\nu'))$. 

To prove the above statement, it suffices to show that $\preceq$ is a simulation relation.

Key Observation: The transition formula for finding the successor of $s'$ and $s$ on any $(t,a) \in \R_{\ge 0} \times \Sigma$ is the same, i.e., $\delta(q, a)$.
By semantics of $\atau$, all the timing constraints appearing in $\delta(q, a)$ that are satisfied by valuation $\nu + t$ will also be satisfied by $\nu'+t$. Hence, every non-deterministic choice that can be taken from $s$ can be taken from $s'$ on $t,a$. More precisely, for every $D = \{(q_1, \nu_1), (q_2, \nu_2), \ldots, (q_j, \nu_j)\} \in \succ^{st}(s, t, a)$ there exists $D' = \{(q_1, \nu'_1), (q_2, \nu'_2), \ldots, (q_j, \nu'_j)\} \in \succ^{st}(s',t,a)$ such that for $1\le i \le j$,  (1)  $\nu_i$ and $\nu'_i$ are constructed from $\nu+t$ and $\nu'+t$ by resetting the same set of clock variables, $X_i \subseteq X$. Hence, $\nu'_i \sim \nu_i$. (2) If $q \in \Qsim$ then either $q_i \in \Qsim$ or $\nu'_i = \nu_i = 0$ (by definition of $\atau$, every transition from a location in $\Qsim$ to a location not in $\Qsim$ should be a strong reset transition). Hence, if $q_i \notin \Qsim$ then $\nu \sim \nu' = 0$, trivially. Hence, $D \preceq D'$. In other words, $\forall D \in \succ^{st}(s,t,a) \exists D' \in \succ^{st}(f^{-1}(s), t, a). D \lasim D'$ (Id 1).

Let $C_{min} \subseteq C$ such that (1) $\forall s' \in C' \exists s \in C_{min}. s \in f(s')$ and (2) $\forall s \in C_{min} \exists s' \in C'. s \in f(s')$ . Notice that such a subset should exist.


By remark \ref{rem:subsetsimulation},  $C \la C_{min}$. Hence, it suffices to show that $C_{min} \la C'$, i.e., Condition (2.1)(2.2) for simulation relation holds. (Condition 2.1) If $C_{min}$ is an accepting state then $C'$ is an accepting state. This is because, if $C_{min}$ is an accepting state then it only contains accepting locations. As all locations in $C'$ also appear in $C_{min}$ (and vicec-versa, all the locations in $C'$ are accepting iff $C_{min}$ is accepting. 
(Condition 2) For any $t, a \in \R_{\ge 0} \times \Sigma$, for any $D \in \succ (C_{min}, t, a)$ there exists $D' \in \succ(C', t, a)$ such that $D \la D'$. This is shown as follows. Let $C_{min}= \{s_1, \ldots s_n\}$ and $C' = \{s'_1, \ldots, s'_m\}$. (Condition 2.2) is equivalent to (Condition 3) For any $t, a \in \R_{\ge 0} \times \Sigma$, $\forall D_1 \in \succ(s_1, t, a). \forall D_2 \in \succ(s_2, t, a).\ldots . \forall D_n \in \succ(s_n, t, a). \exists D'_1 \in \succ(s'_1, t, a). \exists D'_2 \in \succ(s'_2, t, a).\ldots . \exists D'_m \in \succ(s'_m, t, a).$ such that $D_1 \cup D_2 \cup \ldots \cup D_n \preceq D'_1 \cup D'_2 \cup \ldots \cup D'_m$.  
Condition 3 is implied by rearranging,   $\forall D_1 \in \succ(s_1, t, a). \exists D''_{1} \in \succ (f^{-1}(s_1), t, a). D_1 \preceq D_1''  \wedge \forall D_2 \in \succ(s_2, t, a). \exists D''_{2} \in \succ (f^{-1}(s_2), t, a). D_2 \preceq D_2'' \wedge \ldots \wedge \forall D_n \in \succ(s_n, t, a). \exists D''_{n} \in \succ (f^{-1}(s_n), t, a). D_n \preceq D_n''$ which is in turn is implied by (Id 1). Hence, proved.
\end{proof}

\section{Proof of Lemma \ref{lem:main}}
\label{app:main}
\subsection{ATA Semantics as Directed Acyclic Graph}
Equivalent to the Transition Systems based semantics of ATA described in section \ref{sec:ata}, we can also represent runs of ATA $\Aa$ starting from a state $s$ as a Directed Acyclic Graph (DAG) such that each node $m$ of a DAG is labeled with a state $s_m \in S$. All the transitions from nodes at depth $n$ to the nodes at depth $n+1$ are labeled with the same symbol $(t_{n+1},a_{n+1})$. The root node is labeled as $s$. Two children sharing the same parent will have different labels. $M = \{m_1, \ldots m_k\}$ is the set of children of node $m$ iff $C_m =\{s_1, \ldots s_k\}$ is a successor of $s_m$ on time delay $t$ and action $a$, and for all $1\le i \le k$,  $m_i$ is labelled as $s_i$. The set of labels appearing at depth $n$ is the configuration that is reached at step $n$ by transition system $T(\Aa, \{s\})$. See the different colored paths in figure \ref{fig:ta}. Those paths and individual states generates a DAG (a tree in this case). This tree is equivalent to the run presented there.

\subsection{Relation between DAGs of $\Img(R)$}
If $R$ is represented as a run DAG, then $\Img(R)$, is defined similarly. To be precise, $\Img(R)$ is a subDAG of $R$ constructed as follows. At any level $i$ of tree $R$, if there is a node $s'$ such that there exists yet another node $s \ne s'$ and $s \lasim s'$ then we delete the subDAG rooted at $s'$. See an example in figure \ref{example:reduce}. We now prove the main lemma which will be used to show that our construction from $\vwatau$ to Timed Automata terminates. Moreover, the Timed Automata has at most exponentially many locations and polynomially many clocks. 

\begin{figure}[h]
    \includegraphics[scale=0.25]{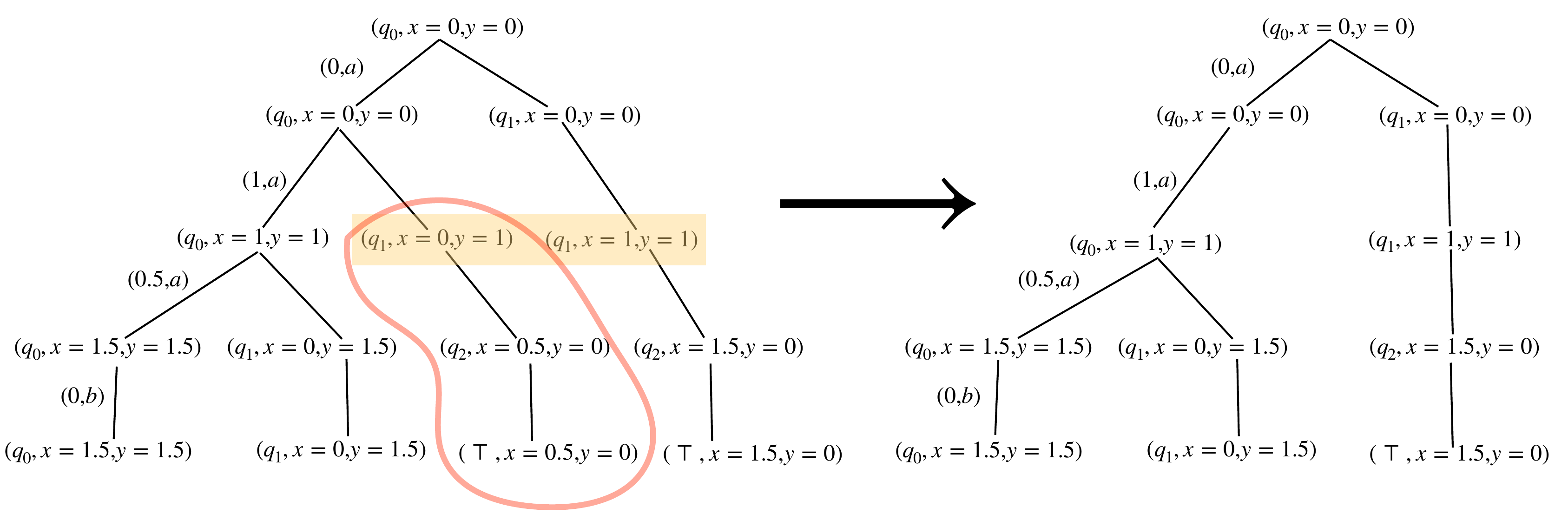}
    \caption{Tree corresponding to run $R$ on $\rhos = (a, 0) (a,1) (a,1.5) (b,1.5)$ for the $\vwatau$ in Fig \ref{fig:running-example}. Remove the highlighted subDAG to get the run DAG for $\Img(R)$}
        \label{example:reduce}
\end{figure}
\subsubsection {Type Sequences of Paths in Run DAG}
Let node $m$ ($m'$) at level $n$ ($n'$) be labelled $(q,\nu)$ ($(q',\nu')$) of a run DAG of $\Aa$. An annotated path between node $n$ and $n'$ can be expressed using a sequence of states visited along the path and the set of clock resets along the path as follows:
\\$P = (q, \nu) [X_1, (q_1, \nu_1)] [X_2, (q_2, \nu_2)] \ldots [X_{m-m'},(q',\nu')]$ where $X_1,\ldots,X_{m-m'} \subseteq X$. 
We associate a type defined by sequence over $H=(2^{X} \times Q) \cup Q^* \cup Q$ to the annotated paths of a run tree of $\vwatau$, $\Aa$, where $Q^{*} = \{q^{*}| q \in Q\}$. 
We define $\type(P)$ as a sequence obtained from $P$ as follows: (1) Drop the clock valuations in the sequence, 
(2) Replace every substring $[X_{i-1}, q_{i-1}][X_{i}, q_i][X_{i+1}, q_{i+1}]\ldots [X_{i+j+1}, q_{i+j+1}]$ with $[X_{i-1}, q_{i-1}][X_{i}, q_i]q_i^* [X_{i+j+1}, q_{i+j+1}]$ iff $q_{i-1} \ne q_{i}$ or $X_{i} \ne \emptyset$, and $q_{i+j+1} \ne q_{i}$ or $X_{i+j+1} \ne \emptyset$, and
$X_{i+1} = X_{i+2} = \ldots = X_{i+j} = \emptyset$ and $q_i = q_{i+1} = q_{i+2} = \ldots = q_{i+j}$  
Notice that if $\Aa$ is a $\vwata$ and $q_i = q_{i+1}$ then $X_{i+1}$ should necessarily be an empty set as self-loops are reset free.  
For example, consider path $P = (q,\nu) [X_1,(q_1, \nu_1)] [\emptyset,(q_1, \nu_1)][\emptyset,(q_1, \nu_1)] [X_2, (q_2, \nu_2)] [X_{3},(q_3,\nu)]$. $\type(P) = q.[X_1,(q_1, \nu_1)]. q_1^*.[X_2, (q_2, \nu_2)].q_2^*.[X_{3},(q_3,\nu)].q_3^*$.

This  sequence $\type (P)$ is called the type of path. Notice that $\type(P)$ is a sequence over $H$. Moreover, as $\Aa$ is a $\vwata$, $|\type(P)|$ is bounded by $|Q|$. 

\todo{Change observation to remark everywhere in the formal proofs.}


\begin{remark}
\label{rem:unique-path-type}
For any $\vwata$, the path type between two nodes labeled $(q,\nu)$ and $(q',\nu')$ is only dependant on $q$ and $q'$. In other words, no matter which run tree we pick and which pair of nodes in that run tree we pick, the type of path between those nodes will only depend on the location appearing in those nodes. This is because the structure of the automaton $\vwata$ is a Tree-like structure.   
\end{remark}

\subsection{Proof}
\textbf{Statament} - Let $\Aa = (Q,\Sigma, \delta, q_0, \Qacc, X, \G)$ be either an 1-$\atau$ or $\vwatau$ . Let Run $R$ of $\Aa$, and $R' = \Img(R) =C_0'' (t_0, a_0)C_1'' (t_1, a_1)\ldots $, then for all $i \ge 1$, $C_i''$ does not contain states $(q,\nu)$ and $(q, \nu')$ where $\nu \ne \nu'$ for any $q \in Q$. In other words, if every location $q \in Q$ appears at most once in any configuration $C_i''$ for any $i \ge 1$. Hence, $|C_i''| \le |Q|$.

\begin{proof}
Notice that if $\Aa$ was 1-$\atau$, the above statement is straightforward as no two clock valuations are incomparable in the case of 1-clock. We now show the same for $\Aa$ being a multi-clock $\vwatau$.

The following statement is equivalent to the statement of the lemma. Let $R$ be any run represented as a run DAG, then $\Img(R)$ is such that at every level $m$, at most one node is labeled with any location $q$.

We now show the following. In the tree corresponding to $\Img(R)$, there is exactly one path $P$ of length $n$ of a particular type (Statement 1). Hence, if there are two paths, $P$ and $P'$ reaching at location $q$ at step $n$, then $\type(P) \ne \type(P')$.  But by remark \ref{rem:unique-path-type} path from a state in location $q'$ to location $q$ is of a unique type. Hence, Statement 1 along with remark \ref{rem:unique-path-type} implies the result.

For paths of length 1, statement 1 trivially holds. Suppose it holds for all the paths of length $\le n-1$ for some $n$.
We now prove the induction step by contradiction. Suppose there are two different paths of the same type $P = (p_1,\nu_1) [Y_2,(p_2,\nu_2)]\ldots [Y_n,(p_n, \nu_n)]$ and $P' = (p'_1,\nu'_1) [Y'_2,(p'_2,\nu'_2)]\ldots [Y'_n,(p'_n, \nu'_n)]$ ending up at states $s = (p_n,\nu_n) =(q,\nu_n)$ and $s' = (p'_n, \nu_n')=(q,\nu_n')$, respectively, at step $n$ such that  $s \not\preceq s'$. Hence, $\nu_n$ and $\nu_n'$ are incomparable. Let $\type(P)=\type(P')=q_0.q_0^*.[X_1,q_1].q_1^*. \ldots.q_{j-1}^*.[X_{j},q_{j}].q_j^*$.  Let $1 < m_1\le m_2\le \ldots \le m_j$ be the points where $P$ takes a non-self-loop transition. That is, $P$ takes a transition from location $q_{i-1}$ to $q_i$ at point $m_i$. Similarly,  $1< m'_1\le m'_2 \le \ldots \le m'_k$ be the points such that $P'$, for any $i\le k$, takes a transition from location $q_{i-1}$ to $q_i$ at point $m'_i$. Without loss of generality, we assume $m_k' \le m_k$.  
Please refer figures \ref{fig:main1}, \ref{fig:main2}, \ref{fig:m}. Following are the possibilities. 
\begin{itemize}
    \item Case 1: $m_k < n$. This implies, $p_{n-1} = p'_{n-1} = p_{n}$.
    Transition from $p_{n-1}$ to $p_{n}$ is a non-reset transition. By condition (2) of $\vwata$, every self-loop transition is a non-reset transition. Hence, if $p_{n-1} = p'_{n-1}$, then the transition from $p'_{n-1}$ to $p'_n$ is a non-reset transition too. Hence, $\nu'_n - \nu_n = \nu'_{n-1} - \nu_{n-1}$. This would imply, $\nu_{n-1}$ and $\nu_{n-1}'$ are incomparable. But paths $P[1..n-1]$ and $P'[1.n-1]$ are of the same type and of length strictly less than $n$. This violates the induction hypothesis.
    \item Case 2: $m'_k = m_k = n$. $p_{n-1} = p'_{n-1} \ne p_n$. By condition (3), all the transitions from any location $q'$ to $q''$ reset the same set of clocks. Hence, the same set of clocks $Y\subseteq X$ will reset while transitioning at step $n$ in both paths. This would again imply, $\nu_{n-1}$ and $\nu_{n-1}'$ are incomparable leading to a contradiction of the induction hypothesis.
    \item Case 3: $m_k'<m_k$ and $m_k = n$. We encourage readers to refer to figure \ref{fig:m}. 
    Suppose $m'_k > m_{k-1}$. This implies that both $P$ and $P'$  are at location $q_{k-1}$ at point $w=m'_k-1$. There are two sub-possibilities, (1)$\nu_{w}$, $\nu'_{w}$ are incomparable. This implies that $P[1..w] \ne P'[1..w]$ but $\type(P[1..w]) = \type(P'[1..w])$, which contradicts the induction hypothesis. In other words, as both $P$ and $P'$ are paths in $\Img(R)$, $\nu_{w}$ and $\nu'_{w}$ should be either equal or incomparable. Otherwise, we should have deleted one of the subtree corresponding to the simulated state from the run tree, discontinuing either $P$ or $P'$. But, $\nu_{w}$ and $\nu'_{w}$ being incomparable violates the induction hypothesis. (2)$\nu_{w} =\nu'_{w}$. This implies that $P[1..w] = P'[1..w]$. 
    In this case, both $P[w..n]$ and $P'[w..n]$ are distinct paths of the same type and length of the subtree rooted at node labelled with state $(q_{k-1}, \nu_w)$ (subgraph of the original tree). As $n-w < n$, this contradicts the induction hypothesis.
    Hence, both sub-possibilities lead to a contradiction.
    Thus the only possibility is that $m'_{k}\le m_{k-1}$. Applying the similar argument, (and let $m_0 = 0$) inductively for all $1 \le i < k$, we have, $m'_{i+1} < m_{i}$. Hence, for any clock $x \in X_1 \cup X_2 \cup \ldots X_k$, $\nu_n(x) \le \nu'_n(x)$ as in $P'$ all these clocks are reset before their corresponding reset in $P$. Moreover, for all the other clocks not in $X_1 \cup \ldots X_k$, the values in $P$ and $P'$ remain the same. Hence, $\nu_n \le \nu'_n$. This contradicts the assumption that $\nu_n$ and $\nu'_n$ are incomparable. Hence proved.

    \begin{figure}
        \scalebox{0.5}{
        \includegraphics{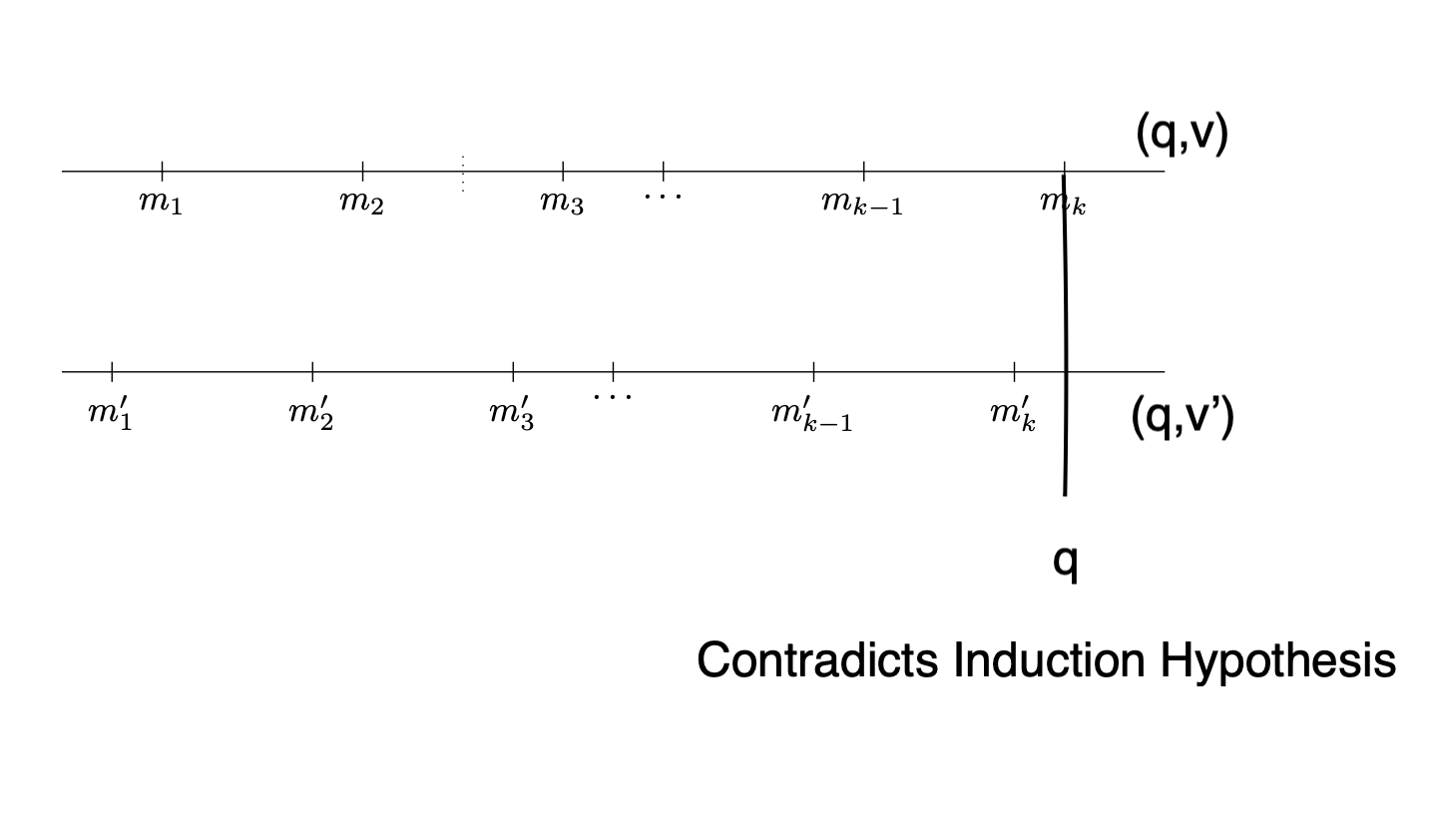}}
        \caption{Figure Corresponding to Case 1 in lemma \ref{lem:main}}
        \label{fig:main1}
    \end{figure}
    
    \begin{figure}
       \scalebox{0.5}{
        \includegraphics{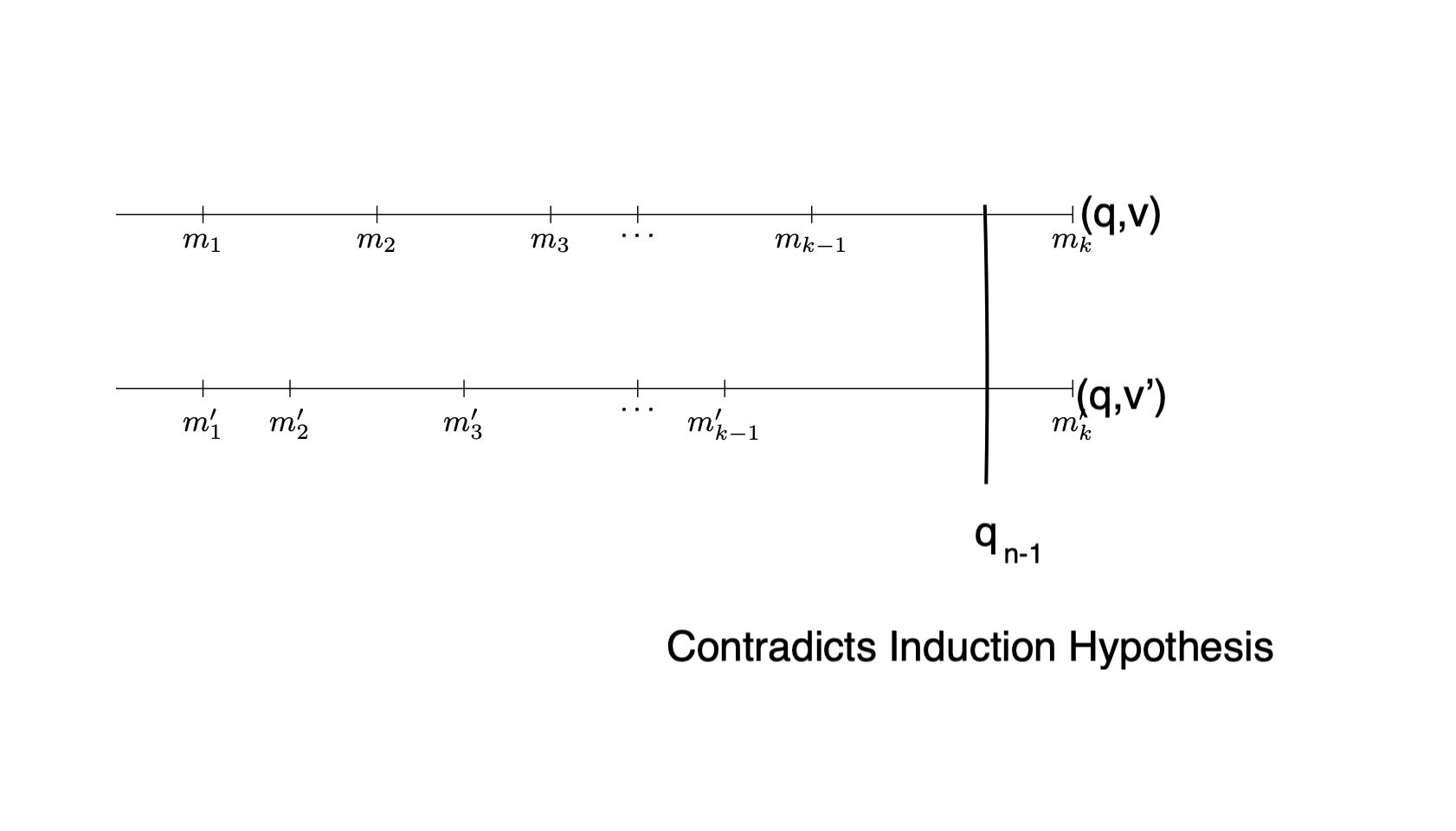}}
        \caption{Figure Corresponding to Case 2 in lemma \ref{lem:main}}
        \label{fig:main2}
    \end{figure}

 \begin{figure}
                \includegraphics[scale=0.2]{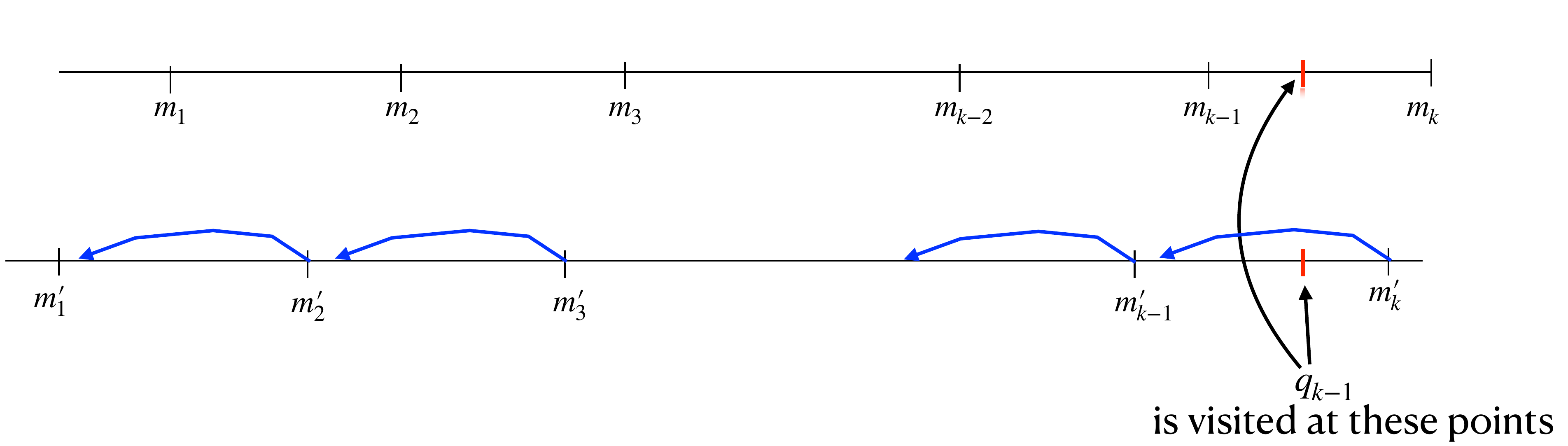}
        \caption{Figure Corresponding to Case 3 in lemma \ref{lem:main}. Note that the point marked above contradicts the Induction Hypothesis. Hence,$m'_k < m_{k-1}$. Recursively arguing about $m'_{k-2}\ldots m'_1$ we get the $\nu' \ge \nu$.}
        \label{fig:m}
    \end{figure}

\end{itemize}

Notice that the argument presented in Case 3 proves a stronger Lemma \ref{lem:main}.  
\end{proof}

\begin{lemma}
    \label{lem:maint}
    Given any pair of paths $P = (q_0,\nu_1) [Y_2,(p_2,\nu_2)]\ldots [Y_n,(p_n, \nu_n)] [Y, q]$ and $P'=(q_0,\nu_1) [Y'_2,(p'_2,\nu'_2)]\ldots [Y'_n,(p'_n, \nu'_n)][Y', q]$ of same length of run tree of $\tred(\Aa)$ such that $p'_n \ne q$ and $p_n = q$, then $\nu'_n \ge \nu_n$.
\end{lemma}
 Intuitively, if from a configuration without where every location appears atmost once, if we fire a transition, leading to 2 states at same location $(q, \nu)$ and $(q, \nu')$, then by $\vwata$ one of them has to be a result of self-loop transition (say $(q, \nu')$), in which case, the clock valuation of the state resulting from self-loop transition $\nu'$ is always greater than or equal to the other valuation $\nu$.     

\section{Formal Construction from $\vwatau$ to NTA}
\label{app:atatonta}
\subsubsection{Formal Construction}
Given a $\vwatau$ or $1-\atau$, $\Aa = (Q, \Sigma, \delta, q_0, \Qacc, X, \G, \Qg, \Ql)$ we get a timed automaton $\A = (\Q, \Sigma, \Delta, q'_0, \Qc, X \times \{0,\ldots, |Q|-1\}, \G)$ and at every step we reduce the size of the location $q \in \Q$ preserving simulation equivalence. Let $V$ be set of all the functions of the form $v: X \mapsto \{0,\ldots, |Q|-1\}$. Let $\lo$ be a set of all the functions from $Q$ to $V\cup \{0\}$. Let $\act$ be a set of all the functions from  $X$ to a sequence (without duplicate) over $\{0,1\ldots |Q-1|\}$. Then $\Q = \lo\times \act$. 

We now give some intuition before formal construction. 
The intuition behind the tuples of $\A$: Our construction technique is inspired by subset construction to eliminate conjunctive transitions in AFA, 
We introduce at most $|Q|$ copies of every clock $x$ of $A$. This is required as there are conjunctive transitions where a clock variable $x$ is reset in one of these transitions while not in another. Hence, there are multiple possible values/copies of variable $x$ that are generated due to these conjunctive transitions. By lemma \ref{lem:main}, the number of clocks that we require at any point of time is bounded by $|Q|$. Hence, the bound $|Q|$ on copies for $X$.

A location of $\A$ is of the form $q' = (L, \ac)$. Intuitively $L(q) = 0$ means that $q'$ doesn't contain the location $q$ of $\Aa$. While, $L(q) = v$, means $q'$ contains the location $q$ of $\Aa$ and the copy of any clock $x \in X$ that is required for computing the successors of $q'$ due to $q$ is given by $(x, v(x))$. Finally, $\ac$ stores all the active copies of $X$. In other words, for any $x \in X$, if $\ac(x) = [1,3,5]$, then it means that clock copies $1,3,5$ are the clock copies that will be required to compute the successor of from location $q'$. Hence, $L(q)(x) \in \{1,3,5\}$ for any $q \in Q$ such that $L(q) \ne 0$.   
And, last reset of $(x,1)$ was before last reset of $(x,3)$ which was before last reset of $(x,5)$. Hence, the value of $(x, 1) \ge (x, 3) \ge (x,5)$.

Initial location is $(L_0, \ac_0)$, such that $L_0(q_0)(x) = 0$ for all $x \in X$ and $L_0(q) = 0$ and for all $q \in Q$, $q \ne q_0$. Moreover, $\ac(x) = [0]$ for all $x \in X$.  

We refer the readers to our running example discussed prior to the construction in the main paper section \ref{sec:exc}. Consider the $\vwatau$, from our running example, figure \ref{fig:running-example}. We show the construction of $\Delta$ for any arbitrary location $q \in \Q$ using the transition function $\delta$. This construction can then be used starting from the initial location $q_0'$ until all the reachable states are explored. Hence, the size of the timed automaton that we get might be much less than the worst case. 

Formal Reduction: Let $q' = (L, \ac)$ be any location in $\Q$. And let $Q' = \{q_1, \ldots, q_n\} = \{q |q \in Q \wedge L(q) \ne 0\}$.  
Without loss of generality, we assume that the output formula of the transition function $\delta$ is given in Disjunctive Normal Form (DNF) \footnote{Notice that conversion to DNF will only incur a blow up in the size of the formula, but not in the size of the locations of the input ATA. Hence, it wouldn't affect the worst-case size in the output NTA.}.  
Hence, for any $q_i \in Q'$, $\delta (q_i,a)=\bigvee \C_i$ where $\C_i$ is a set of clauses (i.e. conjunction of transitions). 
Choose clauses $C_1 \in \C_1, C_2 \in \C_2, \ldots C_n  \in \C_n$ and construct a clause (conjunction of transition) $C = C_1 \wedge C_2 \ldots \wedge C_m$. Let $\C$ be set of all such clauses.
We now show the successor of $q'$ due to a clause $C \in \C$ (recall $C = C_1  \wedge C_2 \ldots C_n $). By conditions 2 and 3 of $\vwata$ definition, every $q \in Q$ appears in at most two of the clauses $C_{i_q}$ and $C_{j_{q}}$ and at most once in each as all the transitions from $i_q$ to $q$ reset same set of clocks (for some $i_{q}, j_{q} \in \{1,2, \ldots, n\}$). Moreover, either $q_{i_q} = q$ or $q_{j_q} = q$. Hence, every location $q \in Q$ appears at most twice in $C$. Let $\appear (q, C) \subseteq \{C_1, C_2, \ldots, C_n\}$ output set of clauses where $q$ appears. Similarly, for any guard $g \in \G$,  let $\appear_G(g, C)\subseteq \{C_1, C_2, \ldots, C_n\}$ output set of clauses where guard $g$ appears.  
Let $C = X_1. Q_1 \wedge X_2. Q_2 \ldots X_m. Q_m \wedge g_1 \wedge g_2 \wedge \ldots \wedge \wedge g_n$ ($Q_1, Q_2, \ldots, Q_m \subseteq Q$). Notice that any $q$ will appear at most once within the scope of a clock reset (by condition 2-3 of $\vwata)$. Notice that $X_1, \ldots X_m \subseteq X$. We abuse the notation by assuming that these sets could be empty too. Hence, $\emptyset.Q_1$ means $\bigwedge Q_1$
Let $\reset(q, C) = X_i$ if $q \in Q_i$.  
For any $x \in X$, let $J_{x} = \{j | 1\le j \le m \wedge x \in X_j\}$. That is, $J_x$ indicates the index of transitions that resets clock $x$. Let $J'_x = \{i | 1 \le i \le n \wedge x$ is reset in all transitions in $C_{i}\}$. Let $J''_x(i) = \{k | 1\le k \le n \wedge L(q_k)(x) = L(q_{i})(x)\}$ indicates indices of the locations in $\{q_1, \ldots q_n\}$ which refer the the same copy of clock $x$ as $q_{i}$.
We now give the step-by-step construction of a transition from $q'$ due to clause $C$. 
\begin{enumerate}
\item $Y_C \leftarrow \emptyset$ (stores the set of clock copies to be reset). $(L_C, \ac_C)$ is the location that is a successor of location $(L, \ac)$ constructed from clause $C$. 
\item For all $q \in Q$ do:
\begin{enumerate}

\item If $q$ doesn't appear in $C$, then $L'(q) = 0$.  

\item If $\appear(q, C) = \{C_{i_q}\}$ and $\reset (q, C) = \emptyset$ then 
\begin{enumerate}
    \item For all $x \in X$: $L_C(q)(x) \leftarrow L(q_{i_q})(x)$ (/*assign the same clock copies as the parent $q_{i_q}$ to the state $q$*/) .
\end{enumerate}
\item If $\appear(q, C) = \{C_{i_q}\}$ and $\reset (q, C) = X_q \ne \emptyset$ then for all $x \in X_q$: 
\begin{enumerate}
    \item If $J_x = \{1,2,\ldots, m\}$ (/* In this case, we need to maintain only one copy of clock $x$ (the $0^{th}$ copy) as all the transitions simultaneously reset their corresponding copies of $x$ to 0.*/)
    $L_C(q)(x) \leftarrow 0$, $\ac_C(x) \leftarrow [0], Y_c \leftarrow Y_c \cup (x,0)$.
    \item  Else if $J_x \subseteq \{1,2,\ldots, m\}$ and $J''_x(i_q) \subseteq J'_x$ (/* In this case, all the transitions exiting from the location $q_k \in Q'$ referring to the same copy of clock $x$ as $q_{i_q}$ reset the clock $x$. Hence, we need not create a new copy of $x$. We just reset the copy of $x$ referred by location $q_{i_q}$ */). Let $w = L(q_{i_q})(x)$. Then,
    $L_C(q)(x) \leftarrow w$, $Y_c \leftarrow Y_c \cup (x, w)$.
    \item Else (/* In this case, the copy of $x$ referred by $q_{i_q}$ is reset by some transitions and not reset by others. Hence, we need to create a new copy of $x$ and reset that copy of $x$ */). Let $w_x$ be the smallest number not appearing in $\ac(x)$.  
    $L_C(q)(x) \leftarrow w_x$, $\ac_C(x) \leftarrow \ac(x).w_x$, $Y_c \leftarrow Y_c \cup (x, w_x)$ where $.$ is concatenation operator.
\end{enumerate}
\item If $\appear(q, C) = \{C_{i_q}, C_{j_q}\}$: (Without loss of generality, assume $q_{j_q} = q$ and $q_{i_q} = q'' \ne q$). (/* Then by lemma \ref{lem:maint} the value of the clock copies associated with $q$ generated from $q''$ will be less than or equal to corresponding values of the clock copies associated with $q$ generated from self-loop. Hence, if $q \in \Ql$ then, if the location $q$ generated from  self-loop reaches the accepting state then the $q$ generated from $q''$ will also reach the accepting state (vice-versa holds in case of $q \in \Qg$)*/).
If  $q \in \Ql$: For all $x \in X$ do: $L_C(q)(x) \leftarrow L(q)(x)$.
If $q \in \Qg$: Do Steps 2(b-c) and continue.
\end{enumerate}
\item Let $G_C \leftarrow \emptyset$ (stores the guards corresponding to the clause $C$). For all $g \in \G$ do:
\begin{enumerate}
\item For all $C_i \in \appear_G(g, C)$ do: $G_c = G_c \cup g^i$, where $g^i=(x, L(q_i)(x)) \in I$ iff $g = x \in I$.  
\end{enumerate}
\item return $Y_C, L_C, \ac_C, G_C$. (/* There is a transition from $q'$ on resetting clocks $Y_C$, asserting guards $G_C$, to a location $(L_C, \ac_C)$*/). 
\end{enumerate}
Hence, $\Delta (q', a) = \bigvee \limits_{C \in \C} (Y_C.(L_C, \ac_C) \wedge \bigwedge G_C)$. Similar reduction works for 1-$\atau$. 
 
The correctness of the above algorithm depends on the following proposition. 
\begin{proposition}
\label{prop:tasim}
    $\A$ to $\tred(A)$ are simulation equivalent. Hence, $L(A)=L(\A)$.
\end{proposition}
To prove the \ref{prop:tasim}, consider the relation $R$ between states of $\A$ and $\tred(A)$ as follows. $((L,\ac), \nu), C) \in R$ iff; $(q, \nu') \in C$ iff $L(q) \ne 0$ and for all $x \in X$, $\nu(x,{L(q)(x)}) = \nu'(x)$. Notice $R$ and $R^{-1}$ are simulation relations from  states of $\A$ to that of $\tred(A)$ and vice-versa, respectively, by inspection of construction.

\subsubsection{Worst Case Complexity:} Notice that the number of bits required to store a state of region automata is the sum of (1) $B_l$ = Number of bits required to store any location $(L,\ac)$, and (2) $B$ = Number of bits required to store the region information. Bits required to store $L = |Q| \times \log(|Q|)\times log(|X|) + 1$ and that required to store $\ac = |X| \times \log(e\times |Q|!) \le |X| \times (|Q| \log |Q| + 1)$ ($e$ is the natural Euler's Number). Hence, $B_l = O(|X| \times \log(|Q|) \times |Q|)$. Moreover, $B = log (|Q| \times |X|)! \times \log(2\times c_{max} + 1)$ where $c_{max}$ is the maximum constant used in the timing constraints in $\G$. Hence, to store a state of region automata, we only need bytes polynomial in the number of locations and clocks of $A$, proving membership of emptiness for 1-$\atau$ and $\vwatau$ in PSPACE. 
Notice that the state containing the location $(L,\ac)$ will only have to store the region information clocks in $\ac$, i.e. active clocks, which, in practice, could be significantly less than the worst case. Hence, lazily spawning clock copies may result in NTA with much less number of clocks than the worst case (i.e. $|X| \times |Q|$).
\end{document}